\documentclass{lmcs}
\pdfoutput=1
\usepackage[utf8]{inputenc}

\usepackage{lastpage}
\lmcsdoi{19}{3}{4}
\lmcsheading{}{\pageref{LastPage}}{}{}%
{Apr.~11,~2022}{Jul.~12,~2023}{}


\usepackage{amsbsy,amssymb,amsmath,braket,pdfsync}
\usepackage{stmaryrd}
\usepackage{xcolor}

\newcommand{\COEFF}{\textsc{Coeff}}
\newcommand{\Prod}{\mathbf{P}}
\newcommand{\CONST}{\textsc{Const}}
\newcommand{\MOD}{\textsc{Mod}}
\newcommand{\bN}{{\mathbb{N}}}
\newcommand{\bZ}{{\mathbb{Z}}}
\newcommand{\cZ}{{\mathcal{Z}}}
\newcommand{\false}{\mathrm{f\!f}}
\newcommand{\lcm}{\mathop{\mathrm{lcm}}}
\newcommand{\true}{\mathrm{t\!t}}
\renewcommand{\ge}{\ensuremath{\geqslant}}
\renewcommand{\le}{\ensuremath{\leqslant}}
\newcommand{\FOMOD}{\mathrm{FO}\bigl[\exists^{(q,p)}x]}
\newcommand{\FOMODmany}{\mathrm{FO}\bigl[\exists^{(t,p)}\bar{x}]}

\newcommand{\FO}{\mathrm{FO}}
\newcommand{\CMOD}{\mathrm{C{+}MOD}}
\newcommand{\C}{\mathrm{FO}[\exists^{\ge c}\bar x,\exists^{=c}\bar x]}
\newcommand{\Cone}{\mathrm{FO}[\exists^{\ge c}x,\exists^{=c}x]}
\newcommand{\qd}{\mathrm{qd}}

\newcommand{\bd}{\mathrm{bd}}
\newcommand{\BD}{\mathrm{BD}}
\newcommand{\fullLogic}{\mathrm{FO}\bigl[\exists^{(t,p)}\bar{x},\exists^{\ge
    c}\bar{x},\exists^{=c}\bar{x}\bigr]}

\renewcommand{\bar}{\overline}

\newcommand{\links}{\bar z_{\mathrm{left}}}
\newcommand{\mitte}{\bar z_{\mathrm{middle}}}
\newcommand{\rechts}{\bar z_{\mathrm{right}}}
\newcommand{\eins}{\bar {z_1}}
\newcommand{\zwei}{\bar {z_2}}
\newcommand{\drei}{\bar {z_3}}
\newcommand{\lex}{\mathrm{lex}}

\begin{document}

\title[Non-unary counting quantifiers for Presburger arithmetic]{On Presburger arithmetic extended\texorpdfstring{\\}{  }with non-unary counting
  quantifiers\rsuper*}

\thanks{The authors thank Christian Schwarz from TU Ilmenau for
  proofreading the paper and for correcting two mistakes. They also
  thank the reviewers of this paper for pointing to some inaccuracies.}
  
\titlecomment{{\lsuper*}This work was partially supported by
    EGIDE/DAAD-Procope TAMTV. It is a considerably extended version of
    the first part of the extended abstract \cite{HabK15} from FoSSaCS
    2015}

\author[P.~Habermehl]{Peter Habermehl\lmcsorcid{0000-0002-7982-0946}}[a]
\author[D.~Kuske]{Dietrich Kuske}[b]

\address{{IRIF, Universit\'e Paris Cit\'e}, France}
\email{Peter.Habermehl@irif.fr}
\address{TU Ilmenau, Germany}
\email{Dietrich.Kuske@tu-ilmenau.de}

\begin{abstract}
  \noindent
  We consider a first-order logic for the integers with addition. This
  logic extends classical first-order logic by modulo-counting,
  threshold-counting and exact-counting quantifiers, all applied to
  tuples of variables (here, residues are given as terms while moduli
  and thresholds are given explicitly). Our main result shows that
  satisfaction for this logic is decidable in two-fold exponential
  space. If only threshold- and exact-counting quantifiers are
  allowed, we prove an upper bound of alternating two-fold exponential
  time with linearly many alternations. This latter result almost
  matches Berman's exact complexity of first-order logic without
  counting quantifiers. 

  To obtain these results, we first translate threshold- and
  exact-counting quantifiers into classical first-order logic in
  polynomial time (which already proves the second result). To handle
  the remaining modulo-counting quantifiers for tuples, we first
  reduce them in doubly exponential time to modulo-counting
  quantifiers for single elements. For these quantifiers, we provide a
  quantifier elimination procedure similar to Reddy and Loveland's
  procedure for first-order logic
  and analyse the growth
  of coefficients, constants, and moduli appearing in this
  process. The bounds obtained this way allow to restrict
  quantification in the original formula to integers of bounded size
  which then implies the first result mentioned above.

  Our logic is incomparable with the logic considered by
  Chistikov et al.\ in 2022.  They allow more general counting operations in
  quantifiers, but only unary quantifiers. The move from unary to
  non-unary quantifiers is non-trivial, since, e.g., the non-unary
  version of the H\"artig quantifier results in an undecidable theory.
\end{abstract}

\maketitle

\section{Introduction}

Presburger arithmetic is the first-order theory of the structure
$\cZ$, i.e., the integers with addition, comparison, binary relations
$\equiv_k$ (standing for equality modulo $k$) for all $k\ge2$, and all
constants $c\in\bZ$. Presburger~\cite{Pre30} developed a quantifier
elimination procedure for this theory and therefore showed its
decidability. The upper complexity bounds of three-fold exponential
time \cite{Opp78} and of two-fold exponential space \cite{FerR79} have
been shown before Berman~\cite{Ber80} proved the exact complexity to
be two-fold exponential alternating time with linearly many
alternations.  Further results in this direction concern the
complexity of fragments of Presburger arithmetic
\cite{RL78,Gra88,Sch97,Haa14}.

Classical first-order logic can be extended by allowing further
quantifiers besides $\exists$ and $\forall$. One such quantifier was
introduced by H\"artig in \cite{Hae62} and is therefore known as
H\"artig quantifier (usually denoted $I$, cf.\ \cite{HerKPV91} for a
survey on this H\"artig quantifier). The formula
$I x\colon\bigl(\varphi(x),\psi(x)\bigr)$ expresses the equality of
the number of witnesses $x$ for $\varphi(x)$ and for $\psi(x)$,
resp. Apelt, in \cite{Ape66}, considered this extension $\FO[I x]$ of
classical first-order logic for the structure of integers with
addition. He provides a system of axioms and derivation rules whose
completeness he proves using a quantifier elimination. Since the
system of axioms and the derivation relation are decidable, he infers
that the $\FO[I x]$-theory of the integers with addition is
decidable. Alternatively, this decidability follows since Apelt's
quantifier elimination is effective and the truth of quantifier free
statements is decidable.

Another possibility of extending classical first-order logic was
considered by Schweikardt \cite{Sch05} who added the
threshold-counting quantifier $\exists^{\ge t}x$ (here, $t$ is a term, $x$
a variable, and the formula $\exists^{\ge t}x\,\varphi$ says ``there
are at least $t$ witnesses $x$ for the formula $\varphi$''); it is not
difficult to see that this extension $\FO[\exists^{\ge t}x]$ is
equally expressive as Apelt's extension $\FO[Ix]$. She provided an
effective quantifier elimination procedure for the quantifier
$\exists^{\ge t}x$ implying the decidability. An alternative
quantifier elimination for $\FO[\exists^{\ge t}x]$ was given by
Chistikov et al.~\cite{ChiHM21,ChiHM22}.

It should be noted that we do not know any elementary upper bounds for
the quantifier elimination procedures from
\cite{Ape66,Sch05,ChiHM21,ChiHM22} for the logics $\FO[Ix]$ and
$\FO[\exists^{\ge t}x]$. Consequently, no elementary upper bounds for
the respective theories of the integers are known.

In our earlier conference paper \cite{HabK15}, we obtained such an
elementary upper bound for the logic $\FO[\exists^{(q,p)}x]$ (here,
$q$ and $p$ stand for natural numbers, $x$ for a variable, and a
formula of the form $\exists^{(q,p)}x\,\varphi$ expresses ``the number
of witnesses $x$ for $\varphi$ is congruent to $q$ modulo $p$''). More
precisely, we presented a quantifier elimination procedure for this
logic, analysed the size of coefficients, constants, and moduli
appearing in the resulting formula, and inferred that quantification
can be bounded to integers of at most triply-exponential absolute
value; as a result, the theory can be decided in doubly exponential
space which matches the best known upper bound for Presburger
arithmetic using deterministic Turing machines. Extending Klaedtke's
automata-based decision procedure for Presburger arithmetic
\cite{Kla08}, our conference paper also contains an automata-based
decision procedure for this logic that runs in triply exponential time
(which is the optimal time bound known for deterministic Turing
machines~\cite{Opp78}).

In \cite{ChiHM22}, Chistikov et al.\ analysed their quantifier
elimination procedure for the logic $\FO[\exists^{\ge t}x]$. For two
fragments, called ``$F$'' and ``monadically guarded PAC'',
respectively, they obtained elementary upper bounds for the decision
problems. The following results follow since the two logics are
contained in the two named fragments.
\begin{itemize}
\item The
  $\FO[\exists^{(t,p)}x,\exists^{\ge c}x,\exists^{=c}x]$-theory of the
  integers is decidable in doubly exponential space (here, $t$ stands
  for a term, $p$ and $c$ for natural numbers, and $x$ for a
  variable).
\item The $\FO[\exists^{\ge c}x,\exists^{=c}x]$-theory of the integers
  can be decided by an alternating Turing machine using doubly
  exponential time and linearly many alternations (in this logic, no
  modulo-counting quantifiers are allowed and the thresholds are given
  explicitly). More precisely, the number of alternations is not only
  bounded by the \emph{length}, but even by the \emph{depth} of the
  formula.
\end{itemize}
These two upper bounds coincide with the best known upper bounds for
Presburger arithmetic wrt.\ deterministic and alternating Turing
machines, resp.

It should be noted that all logics considered so far extend classical
first order logic by \emph{unary} quantifiers, i.e., the quantifiers
$I$, $\exists^{\ge t}$, and $\exists^{(t,p)}$ bind a single
variable. They can easily be extended to bind tuples of variables,
e.g., the formula
\[
  I(x',y')\colon\bigl((x'=0\land 0\le y'<z),(0\le x'<x\land 0\le
  y'<y)\bigr)
\]
expresses that the number $z$ of pairs $(0,y')$ satisfying $0\le y'<z$
equals the number $x\cdot y$ of pairs $(x',y')$ satisfying $0\le x'<x$
and $0\le y'<y$, i.e., $z=x\cdot y$. Hence, allowing this non-unary
H\"artig quantifier $I\bar{x}$ leads to an undecidable theory, the
resulting logic $\FO[I\bar{x}]$ does not possess effective quantifier
elimination, and the same applies for the non-unary version of the
threshold-counting quantifier $\exists^{\ge t}\bar{x}$. Chistikov et
al.\ ask in the introduction of \cite{ChiHM21} whether non-unary
counting quantifiers $\exists^{\ge c}\bar{x}$ and
$\exists^{=c}\bar{x}$ lead to (efficiently) decidable theories.  In
this paper, we answer this question in the affirmative proving that
the non-unary versions of the quantifiers $\exists^{(t,p)}\bar{x}$,
$\exists^{\ge c}\bar{x}$, and $\exists^{\ge c}\bar{x}$ (where the
threshold is given explicitly) behave much better than H\"artig's
quantifier $I$. Namely, we prove the two complexity bounds that follow
from the work by Chistikov et al.\ on the fragments ``$F$'' and
``monadically guarded PAC'' also for the non-unary quantifiers:
\begin{itemize}
\item The
  $\FO[\exists^{(t,p)}\bar{x},\exists^{\ge
    c}\bar{x},\exists^{=c}\bar{x}]$-theory of the integers is
  decidable in doubly exponential space (here, $t$ stands for a term,
  $p$ and $c$ for natural numbers, and $\bar{x}$ for a tuple of
  variables).
\item The $\FO[\exists^{\ge c}\bar{x},\exists^{=c}\bar{x}]$-theory of
  the integers can be decided by an alternating Turing machine using
  doubly exponential time and linearly many alternations. As opposed
  to the above mentioned result on the unary versions of these
  quantifiers, we cannot prove that the number of alternations is 
  bounded by the depth of the formula.
\end{itemize}
Despite the similarity of results, we cannot follow the route of proof
used by Chistikov et al.\ since they start from their handling of the
unary H\"artig quantifier which cannot be extended to its non-unary
version. Differently, we proceed as follows.
\begin{enumerate}
\item In polynomial time, we compute from a formula in the full
  logic $\fullLogic$ an equivalent formula in the fragment
  $\FOMODmany$, that is, non-unary threshold- and exact-counting
  quantifiers can be eliminated in polynomial time. This procedure
  does not intro\-duce new modulo-counting quantifiers; consequently,
  from a formula from $\C$, it computes an equivalent formula from
  classical first-order logic $\FO$. Since the ``block depth'' (a
  notion defined later, it is bounded by the length of the formula) of
  the resulting formula is linear in the size of the original one, we
  obtain that the satisfaction relation for $\C$ is decidable in
  two-fold exponential alternating time with $O(n)$ many
  alternations. Note that this is very close to Berman's optimal
  result for $\FO$ where only $n$ alternations are necessary
  \cite{Ber80}.
\item We provide a quantifier elimination procedure for the logic
  $\FOMODmany$ and therefore, by the first result, for
  the full logic $\fullLogic$.  It follows that this full logic agrees
  in expressive power with classical first-order logic $\FO$.
\item Analysing the size of constants, coefficients, and moduli
  appearing in this procedure, we can restrict quantification to
  integers of bounded size. As a result, we get a decision procedure
  in two-fold exponential space for the full logic $\fullLogic$. Note
  that this equals the best known upper bound using Turing machines
  for classical first-order logic $\FO$ from~\cite{FerR79}.
  
\end{enumerate}

\section{Preliminaries}
We consider $0$ a natural number.

\paragraph{The structure}
The universe of the structure $\cZ$ is the set of integers $\bZ$. On
this set, we consider the constants $c\in\bZ$, the binary
function~$+$, the binary relation~$<$, and the binary
relations~$\equiv_k$ for $k\ge2$ (with $m\equiv_k n$ iff $k$ divides $m-n$).

\paragraph{Terms and assignments}
We will use the countable set $\{x_i\mid i\in\bN\}$ of variables.
\emph{Terms} are defined by induction: $x_i$ and $c$ are terms for
$i\in\bN$ and $c\in\bZ$, and $as$ and $s+t$ are terms whenever
$a\in\bZ$ and $s$ and $t$ are terms (we write $-s$ for the term
$(-1)\cdot s$ and $s-t$ for $s+(-1)\cdot t$).

An \emph{assignment} is a function $f\colon\{x_i\mid i\in\bN\}\to\bZ$
that assigns integers to variables. In a natural way, an assignment
$f$ is extended to a function (also denoted $f$) that maps terms to
integers. Two terms $s$ and $t$ are \emph{equivalent} if $f(s)=f(t)$
holds for all assignments~$f$; we write $s\Leftrightarrow t$ to denote
that $s$ and $t$ are equivalent.\footnote{Usually, one writes
  $s\equiv t$ for the equivalence of terms and formulas, but this
  might lead to confusion in this paper because of the central role of
  the relations $\equiv_k$ for $k\ge2$.}

A term is in \emph{normal form} if it is of the form
$t'=\bigl(\cdots(a_1x_{i_1}+a_2x_{i_2})+\cdots a_n x_{i_n}\bigr)+c$
with $i_1<i_2<\cdots <i_n$, $a_1,\dots,a_n\neq 0$, and $c\in\bZ$. Note
that, for any term $t$, there exists a unique equivalent term in
normal form. For a term $t$ with normal form $t'$, we call $a_j$ the
coefficient of $x_{i_j}$ and $c$ the constant; note that coefficients
are non-zero, but the constant can be zero (in which case we call the
term $t$ \emph{constant-free}).

If the normal form of a term $t$ does not contain the variable $x_i$,
then we call $t$ an \emph{$x_i$-free term}.

\paragraph{Atomic formulas}

Expressions of the form $s<t$ (also written $t>s$) and $s\equiv_k t$
for terms $s$ and $t$ and a natural number $k\ge1$ are called
\emph{atomic formulas}. We extend an assignment $f$ to a function
(also denoted $f$) that maps atomic formulas to the truth values
$\true$ and $\false$: $f(s<t)=\true$ iff $f(s)<f(t)$ and
$f(s\equiv_k t)=\true$ iff $k$ divides $f(s)-f(t)=f(s-t)$.  Two atomic
formulas $\alpha$ and $\beta$ are \emph{equivalent} if
$f(\alpha)=f(\beta)$ holds for all assignments $f$; we write
$\alpha\Leftrightarrow\beta$ for this fact.

Let $x$ be a variable. An atomic formula $\varphi$ is
\emph{$x$-separated} if there are an $x$-free term~$t$ and a
non-negative integer $a\in\bN$ such that $\varphi$ is of the form
$ax < t$, $t < ax$, or $ax \equiv_k t$. If $t$ is an $x$-free term,
then, e.g., the formula $0x\equiv_k t$ is $x$-separated. Since $0$ is
the normal form of $0x$, also the formulas $0\equiv_kt$, $0<t$, and
$t<0$ are considered to be $x$-separated (despite the fact that it
does not mention $x$ at all).  It follows that, for any atomic formula
$\alpha$ and any variable $x$, there exists an equivalent
$x$-separated atomic formula.

An atomic formula is \emph{constant separated} if it is of the form
$c < s$, $s<c$, or $s \equiv_k c$ where $s$ is a constant-free term
and $c\in\bZ$ a constant. Again, for any atomic formula $\alpha$,
there exists an equivalent constant separated atomic formula.

\paragraph{Formulas}

Formulas of classical first-order logic are built from atomic formulas
using the quantifier~$\exists$ (applied to single variables) and the
Boolean combinators negation, conjunction, implication, and
equivalence. We extend this classical logic by 
quantifiers that allow \emph{threshold-} ($\exists^{\geq c}$)
and \emph{exact-counting} ($\exists^{=c}$) as
well as \emph{modulo counting} ($\exists^{(t,p)}$), all applied to \emph{tuples} of
variables.
\begin{defi}
  \emph{Formulas} of the logic $\fullLogic$ are defined by induction:
  \begin{enumerate}[(1)]
  \item Any atomic formula is a formula.
  \item If $\varphi$ and $\psi$ are formulas, then so are
    $\lnot\varphi$, $\varphi\land\psi$, $\varphi\lor\psi$,
    $\varphi \to \psi$ and
    $\varphi\leftrightarrow\psi$.
  \item If $\varphi$ is a formula and $y$ a variable, then
    $\exists y\colon\varphi$ is a formula.
  \item If $\varphi$ is a formula, $t$ a term, $y_1,\dots,y_\ell$ (with
    $\ell\ge1$) distinct variables, and $p\ge 2$ a natural number, then
    $\exists^{(t,p)}(y_1,\dots,y_\ell)\colon\varphi$ is a formula.
  \item If $\varphi$ is a formula, $y_1,\dots,y_\ell$ (with
    $\ell\ge1$) are distinct variables, and
    $c\ge1$ is a  natural number, then
    $\exists^{\ge c} (y_1,\dots,y_\ell)\colon\varphi$ and
    $\exists^{=c} (y_1,\dots,y_\ell)\colon\varphi$ are formulas.
  \end{enumerate}
  The size $|\varphi|$ of a formula $\varphi$ is the amount of space needed to
  write it down (we assume integers to be written in binary and
  variables to have size one).
\end{defi}

For certain fragments of the logic
$\fullLogic$
we use the following naming scheme.
\begin{itemize}
\item $\FO[\cdots,\exists^{(t,p)}\bar{x}\cdots]$ denotes that item (4)
  can be used in the construction of formulas without any
  restriction. $\FO[\cdots,\exists^{(q,p)}\bar{x}\cdots]$ limits the
  use of item (4) to the case that $t$ is a constant from $\bN$ (and
  not an arbitrary term), and $\FO[\cdots,\exists^{(q,p)}x\cdots]$
  requires, in addition, that (4) is only used with $\ell=1$, i.e., we
  can use the unary modulo-counting quantifiers with
  constant residue, only.
\item $\FO[\cdots,\exists^{\ge c}\bar{x},\exists^{=c}\bar{x}\cdots]$
  denotes that item (5) can be used in the construction of formulas
  without any restriction. Similarly to the above,
  $\FO[\cdots,\exists^{\ge c}x,\exists^{=c}x\cdots]$ restricts the
  use of item (5) to the case $\ell=1$, i.e., we can use the unary
  threshold- and exact-counting quantifiers, only.
\end{itemize}

\begin{rem}
  The logics $\Cone$, $\FOMOD$, and
  $\FO[\exists^{(q,p)}x,\exists^{\ge c}x,\exists^{=c}x]$ are often
  denoted $\mathrm{C}$, $\mathrm{FO}{+}\mathrm{MOD}$ and $\CMOD$,
  respectively.
\end{rem}
We  can further extend an assignment $f$ in the
standard way to a function (also denoted~$f$) that maps formulas to the
truth values $\true$ and $\false$. 

Before we define the semantics of quantified formulas, we need the
following definitions. For $\ell\ge1$, $\bar y=(y_1,\dots,y_\ell)$ an
$\ell$-tuple of distinct variables, and
$\bar a=(a_1,\dots,a_\ell)\in\bZ^\ell$, we let $f_{\bar y/\bar a}$ be
the assignment that maps the variable $y_i$ to the value $a_i$ (for
all $1\le i\le\ell$) and, apart from this, coincides with the
assignment~$f$. In other words, $f_{\bar y/\bar a}(y_i)=a_i$ for all
$1\le i\le\ell$ and $f_{\bar y/\bar a}(x)=f(x)$ for all variables
$x\notin\{y_1,\dots,y_\ell\}$.

To define the semantics of the quantifiers, let $\varphi$ be a
formula, $t$ a term, $ y_1,\dots,y_\ell$ distinct variables, $p\ge2$,
and $c\ge1$. With $\bar y=(y_1,\dots,y_\ell)$, we then define the
following:
\begin{itemize}
\item $f\bigl(\exists y_1\colon \varphi\bigr)=\true$ iff
   there exists $a\in\bZ$ such that
  $f_{y_1/a}(\varphi)=\true$.
\item $f\bigl(\exists^{(t,p)} \bar y\colon \varphi\bigr)=\true$ iff
  the set $\{\bar a\in\bZ^\ell\mid f_{\bar y/\bar a}(\varphi)=\true\}$
  is finite and
    \[
      \Bigl|
       \bigl\{\bar a\in\bZ^\ell\colon f_{\bar y/\bar a}(\varphi)=\true\bigr\}
      \Bigr|
      \equiv_p f(t)\,.
    \]
  In other words, the formula $\exists^{(t,p)}\bar y\colon\varphi$
  expresses that the number of witnessing tuples $\bar{y}$ for
  $\varphi$ is (modulo~$p$) congruent to the value of the term $t$.
\item $f\bigl(\exists^{\ge c} \bar y\colon \varphi\bigr)=\true$ iff
    \[
      \Bigl|
       \bigl\{\bar a\in\bZ^\ell\colon f_{\bar y/\bar a}(\varphi)=\true\bigr\}
      \Bigr|
      \ge c\,.
    \]
  In other words, the formula $\exists^{\ge c}\bar y\colon\varphi$
  expresses that the number of witnessing tuples $\bar{y}$ for
  $\varphi$ is at least~$c$ (and possibly infinite). With $\ell=1$, $\exists^{\ge1}$ is the
  usual existential quantifier $\exists$. This easy observation allows
  us to consider $\exists$ as an abbreviation and therefore to skip
  item (3) in the definition of fragments of the full logic
  $\fullLogic$, provided item (5) is allowed with $\ell=1$.
\item $f\bigl(\exists^{=c} \bar y\colon \varphi\bigr)=\true$ iff
  $\bigl|\{\bar a\in\bZ^\ell\mid f_{\bar y/\bar
    a}(\varphi)=\true\}\bigr|=c$.  
\end{itemize}

Two formulas $\alpha$ and $\beta$ are \emph{equivalent} if
$f(\alpha)=f(\beta)$ holds for all assignments $f$; we write
$\alpha\Leftrightarrow\beta$ for this fact.

Clearly, the formula $\exists^{=c}\bar{y}\colon\varphi$ is equivalent
to
$\exists^{\ge c}\bar{y}\colon\varphi\land\lnot\exists^{\ge
  c+1}\bar{y}\colon\varphi$, i.e., we can eliminate any occurrence of
$\exists^{=c}$ without changing the semantics of a formula. But this
elimination may increase the size of the formula exponentially.

Note that $f(s<t\lor s>t)=\true$ iff $f(s)\neq f(t)$ since $<$ is a
strict linear order on the set $\bZ$. Therefore, we will write $s=t$
as abbreviation of the formula $\lnot(s<t\lor s>t)$. Similarly,
$s\le t$ stands for $\lnot s>t$ and sequences of comparisons like
$s_1\le s_2\le s_3$ denote the conjunction
$s_1\le s_2\land s_2\le s_3$. Similarly, we write $\forall x\,\varphi$
as abbreviation for $\lnot\exists x\,\lnot\varphi$.

We define the quantifier-depth $\qd(\varphi)$ of formulas
$\varphi\in\fullLogic$ by induction:
\begin{itemize}
\item If $\varphi$ is an atomic formula, then $\qd(\varphi)=0$.
\item If $\varphi=\lnot\alpha$, then $\qd(\varphi)=\qd(\alpha)$.
\item If
  $\varphi\in\{\alpha\land\beta, \alpha\lor\beta, \alpha\to\beta,
  \alpha\leftrightarrow\beta\}$, then
  $\qd(\varphi)=\max\bigl\{\qd(\alpha),\qd(\beta)\bigr\}$.
\item If $\varphi=\exists x\colon\alpha$, then
  $\qd(\varphi)=1+\qd(\alpha)$.
\item If $\varphi$ is any of the formulas
  $\exists^{\ge c}(y_1,\dots,y_\ell)\colon\alpha$,
  $\exists^{=c}(y_1,\dots,y_\ell)\colon\alpha$, or
 $\exists^{(t,p)}(y_1,\dots,y_\ell)\colon\alpha$,
  then $\qd(\varphi)=\ell+\qd(\alpha)$.
\end{itemize}
Note that the quantifier depth depends on the length of tuples of
variables that follow a quantifier, i.e., it increases by $\ell$
whenever we prepend a quantifier $\exists^{\dots}(y_1,\dots,y_\ell)$ to
a formula.

The overall goal of this paper is to obtain an elementary decision procedure for
the full logic $\fullLogic$. As a first step, we will transform a
formula $\alpha$ from $\fullLogic$ into an equivalent formula $\beta$
from $\FOMOD$, that will later be transformed into an
equivalent quantifier-free formula~$\gamma$. To control the form of
the resulting formulas $\beta$ and $\gamma$, we define the following
sets.

\begin{defi}
  Let $\varphi \in \fullLogic$ be a formula. Then
  $\COEFF(\varphi)\subseteq\bZ$ is the set of integers $0,\pm1,\pm2$
  and $\pm a$ where $a$ is a coefficient in the term $s_1-s_2$ for some
  atomic formula $s_1<s_2$ from~$\varphi$. Similarly,
  $\CONST(\varphi)\subseteq\bZ$ is the set of integers $0,\pm1$, $\pm2$, and
  $\pm c$ where $c$ is the constant term in $s_1-s_2$ for some atomic
  formula $s_1<s_2$ from $\varphi$.

  The set $\MOD(\varphi)\subseteq\bN$ contains $1$ and all integers
  $k\ge1$ such that an atomic formula of the form $s_1\equiv_k s_2$ or
  some quantifier $\exists^{(t,k)}$ appears in~$\varphi$. Finally,
  $\Prod(\varphi)=\COEFF(\varphi)\cup\MOD(\varphi)$.
\end{defi}

\begin{exa}
  Consider the following formula $\varphi$:
    \begin{align*}
       & \exists^{(17x+25,23)}(y_1,y_2)\colon 2y_1<3y_2\land 4y_2<56\\
      \land
       & \exists^{\ge343}y\colon {-}13x+2<3x+y-2\land 57x\equiv_{13}2y+27
    \end{align*}
  Then we have
    \begin{align*}
      \COEFF(\varphi)&=\{0,\pm1,\pm2,\pm3,\pm4,\pm16\}\,,\\
      \CONST(\varphi)&=\{0,\pm1,\pm2,\pm56,\pm4\}\,\\
      \MOD(\varphi)&=\{1,13,23\}\,,\text{ and}\\
      \Prod(\varphi)&=\{0,\pm1,\pm2,\pm3,\pm4,\pm16,23,13\}\,.                 
    \end{align*}
\end{exa}
Note that $\COEFF(\varphi)$ and $\CONST(\varphi)$ depend on
subformulas of the form $s<t$, but not on subformulas of the form
$s\equiv_k t$. On the other hand, $\MOD(\varphi)$ only depends on
subformulas of the form $s\equiv_k t$ and on the moduli $k$ of
modulo-counting quantifiers $\exists^{(t,k)}$ appearing in $\varphi$.

\subsection{An excursion into Presburger arithmetic}

Berman proved in \cite{Ber80} that Presburger arithmetic is complete
for the class STA$(*,2^{2^{O(n)}},n)$ of all problems that can be
solved by an alternating Turing machine in doubly exponential time
with $n$ alternations. Here, we are mainly interested in the proof of
the upper bound. He presents this proof in a very sketchy way
essentially saying that Ferrante and Rackoff have shown in
\cite{FerR79} that quantification can be reduced to integers of at
most triply exponential size (which can be represented in doubly
exponential space). It should be noted that this latter result holds
for any formula, no matter whether it is in prenex normal form or it
contains the Boolean connective $\leftrightarrow$. Berman's result
actually means that the algorithm by Ferrante and Rackoff can be
implemented on an alternating Turing machine with the above time and
alternation bound. Looking into the algorithm from \cite{FerR79}, one
sees that the formula is first transformed into prenex normal form and
that then, the alternation of the Turing machine equals the quantifier
alternation depth of the resulting formula. Note that turning a
formula into prenex normal form is possible in polynomial time
whenever the Boolean connectives are restricted to $\lnot$, $\lor$,
$\land$, and $\to$. Differently here, we also allow the connective
$\leftrightarrow$ which gives a convenient way to write certain
formulas succinctly. But in the presence of this connective, we do not
know how to compute equivalent formulas in prenex normal form in
polynomial time.

For later reference, we now sketch a proof that, also in the presence
of $\leftrightarrow$, Berman's upper bound holds. Since the
computation of prenex normal forms is too costly, we need another
bound for the alternation. To this aim, we define the \emph{block
  depth} of a formula. Intuitively, the block depth $\bd^{\FO}(\alpha)$ of
the formula $\alpha\in\FO$ bounds the number of blocks of
existential quantifiers along any path in the syntax tree of $\alpha$.

\begin{defi}\hfill
  \begin{itemize}
  \item $\BD^{\FO}_0$ is the set of atomic formulas.
  \item For $n\ge1$, the set $\BD^{\FO}_n$ contains the formulas of the form
    $\exists x_1\,\exists x_2\,\dots\,\exists x_m\colon\beta$ where
    $m\ge0$ and $\beta$ is a Boolean combination (possibly using
    $\lnot$, $\land$, $\lor$, $\to$, and $\leftrightarrow$) of
    formulas from $\BD^{\FO}_{n-1}$.
  \item The \emph{block depth} $\bd^{\FO}(\alpha)$ of a formula $\alpha \in \FO$ is
    the minimal natural number $n$ with $\alpha\in\BD^{\FO}_n$.
  \end{itemize}
\end{defi}
Note that the block depth of any formula is at most half of its depth
(which is the maximal length of a branch in the syntax tree) and
therefore half of its length.

With this definition in place, we can now formulate Berman's upper
bound for first-order logic in presence of the Boolean connective
$\leftrightarrow$.

\begin{thm}\label{T-Berman}
  There is an alternating Turing machine that, on input of a closed
  formula $\varphi\in \FO$, decides in time doubly exponential in
  $|\varphi|$ with $2\,\bd^{\FO}(\varphi)\le|\varphi|$ alternations whether
  $\varphi$ holds or not.
\end{thm}

\begin{proof}[Proof sketch]
  The alternating algorithm runs as follows:
  \begin{itemize}
  \item If $\varphi$ is atomic, then validity of the closed formula
    $\varphi$ is checked deterministically.
  \item Now let $\varphi=\exists x_1\dots \exists x_m\colon\psi$ where
    $\psi$ is a Boolean combination of formulas
    $\sigma_1,\dots,\sigma_\ell$ of block depth at most $n$. Then the
    alternating algorithm first guesses $m$ integers $k_1,\dots,k_m$
    of bounded size (which suffices by \cite{FerR79}) as well as a set
    $X\subseteq\{1,2,\dots,\ell\}$. Then, it branches universally
    checking that
    \begin{enumerate}
    \item the Boolean combination $\psi$ holds while assuming $X$ is the set
      of indices $i$ such that $\sigma_i(k_1,\dots,k_m)$ holds,
    \item for all $i\in X$, the closed formula
      $\sigma_i(k_1,\dots,k_m)$ holds, and
    \item for all $j\in \{1,2,\dots,\ell\}\setminus X$, the closed formula
      $\sigma_j(k_1,\dots,k_m)$ does not hold.
    \end{enumerate}
    Thus, the algorithm first branches existentially and then
    universally before checking whether the corresponding formulas
    $\sigma_i$ of block depth $\le n$ hold or not. \qedhere
  \end{itemize}
\end{proof}

\section{Existential and unary modulo-counting quantifiers suffice}

In this section, we will transform a formula from $\fullLogic$ into an
equivalent one from $\FOMOD$. Note that the logic $\FOMODmany$ is an
intermediate logic between these two logics:
\begin{itemize}
\item In the logic $\fullLogic$, we can use the non-unary
  threshold- and exact-counting quantifiers
  $\exists^{\ge c}(x_1,\dots,x_\ell)$ and
  $\exists^{=c}(x_1,\dots,x_\ell)$ while $\FOMODmany$ does not
  allow threshold- and exact-counting quantification.
\item $\FOMODmany$ allows non-unary modulo-counting
  quantifiers $\exists^{(t,p)}(x_1,\dots,x_\ell)$ with $t$ an
  arbitrary term while $\FOMOD$ allows only unary
  modulo-counting quantification of the form $\exists^{(q,p)}x$ with
  $q\in\bN$.
\end{itemize}
We will transform a formula from $\fullLogic$ first into an equivalent
formula from $\FOMODmany$, i.e., we will eliminate threshold counting
quantifiers. In a second step, the resulting formula from $\FOMODmany$
will be translated into an equivalent one from $\FOMOD$, i.e., we will
eliminate non-unary modulo-counting quantifiers as well as terms as
residue. Both these transformations will leave the sets of
coefficients, constants, and moduli unchanged; the first
transformation will be done in polynomial time while the second one
uses doubly exponential time.

\subsection{Elimination of threshold- and exact-counting quantifiers}
  
Here, we give the transformation from $\fullLogic$ to $\FOMODmany$.
We will provide a polynomial-time transformation that does not change
the sets $\COEFF$, $\CONST$, and $\MOD$. In addition, this
transformation will not introduce new modulo-counting quantifiers so
that formulas from $\C$ get translated into equivalent formulas
$\varphi$ from first-order logic\footnote{Stefan G\"oller (private
  communication) explained to us a polynomial translation of formulas
  from $\FO[\exists^{\ge c}x]$ to $\FO$. The work in this section is
  an extension and elaboration of his idea.} whose validity can then
be checked using Theorem~\ref{T-Berman}.

We now come to the translation, i.e., to the elimination of threshold-
and exact-counting quantifiers for tuples. First note that the
formulas $\exists^{=c}\bar y\colon\varphi$ and
$\exists^{\ge c}\bar y\colon\varphi\land\lnot\exists^{\ge c+1}\bar
y\colon\varphi$ are clearly equivalent, i.e., semantically, there is
no need for the exact-counting quantifier $\exists^{=c}$. But applying
this replacement to all exact-counting quantifiers in a formula
increases the size of the formula exponentially. Similarly,
$\exists^{\ge c}\bar y\colon\varphi$ is equivalent to
  \begin{align*}
    \exists \bar{y_1}\,\exists \bar{y_2}\,\dots\,\exists
    \bar{y_c}\colon & \bigwedge_{1\le i<j\le c} \lnot \bar{y_i}=\bar{y_j}\\
                    &\land \forall \bar y\colon
                      \biggl(\Bigl(\bigvee_{1\le i\le c}\bar y=\bar{y_i}\Bigr)\to
  \varphi\biggr)
  \end{align*}
(where $(y^1,y^2,\dots,y^\ell)=(y_i^1,\dots,y_i^\ell)$ abbreviates
$\bigwedge_{1\le j\le \ell}y^j=y^j_i$). Since the constant $c$ is
written in binary, already the prefix of existential quantifiers is of
exponential length, i.e., also this transformation incurs an
exponential blow-up in formula size. Finally note that the non-unary
quantifiers $\exists \bar y$ and $\forall \bar y$ are equivalent to
$\exists y^1\,\exists y^2\cdots\exists y^{\ell}$ and
$\forall y^1\,\forall y^2\cdots\forall y^{\ell}$, respectively.

Thus, we saw that any formula from $\fullLogic$ can be transformed
into an equivalent one from $\FOMODmany$ (and similarly for $\C$ and
$\FO$), but at the cost of an exponential size increase.  Our first
result shows that this size increase can be avoided. 

The crucial part in this construction is the elimination of a
threshold- or exact-counting quantifier in front of a formula from
$\FOMODmany$ or from $\FO$, respectively. This construction adapts a
binary search strategy. For instance, the formula
$\exists^{=2c} y\colon y_0\le y<y_1\land\varphi(y)$ expresses
that the interval\footnote{All intervals in this paper are
  considered as sets of integers or of tuples of integers.}
$[y_0,y_1)$ contains precisely $2c$ many numbers $y$ satisfying
$\varphi$. This is equivalent to saying that there exists some number
$y_{\frac{1}{2}}$ in the said interval such that both intervals
$[y_0,y_{\frac{1}{2}})$ and $[y_{\frac{1}{2}},y_1)$ contain precisely $c$
numbers satisfying $\varphi$. The constructed formula then contains
the conjunction of the two formulas
$\exists^{= c}y\colon \bigl(y_0\le y<y_{\frac{1}{2}}\land
\varphi(y)\bigr)$ and
$\exists^{=c}y\colon \bigl(y_{\frac{1}{2}}\le y<y_1\land
\varphi(y)\bigr)$. Therefore, using this binary-search idea alone does
not prevent an exponential blow-up.  The solution is to replace the
conjunction of these two formulas by an expression of the form
\[
   \forall
  a,b\colon \Bigl((a,b)\in\{(y_0,y_{\frac{1}{2}}),(y_{\frac{1}{2}},y_1)\}\to \exists^{=c} y\colon
  \bigl(a\le y<b\land \varphi(y)\bigr)\Bigr)\,.
\]
This idea (known as Fischer-Rabin-trick) goes back to \cite{FisR74}
where it is attributed to earlier work by Fischer and Meyer as well
as by Strassen without specifying concrete publications.

A similar idea transforms the formula
$\exists^{=2c+1}y\colon y_0\le y<y_1\land\varphi(y)$ into 
\[
  \exists y_{\frac{1}{2}}\colon
  \begin{array}[t]{cl}
    & y_0<y_{\frac{1}{2}}<y_1\\
    \land & \varphi(y_{\frac{1}{2}})\\
    \land & \forall a,b\colon
            \Bigl((a,b)\in\{(y_0,y_{\frac{1}{2}}),(y_{\frac{1}{2}}+1,y_1)\}\to
            \exists^{=c} y\colon \bigl(a\le y<b\land
            \varphi(y)\bigr)\Bigr)\,.
  \end{array}
\]
Note that this results in an exponential increase in formula size
since the formula $\varphi$ is mentioned twice. To avoid this size
increase, we ``postpone'' the evaluation of the formula
$\varphi(y_{\frac{1}{2}})$. Slightly more precisely, the above
construction proceeds recursively since in both cases, we have the
subformula
$\exists^{=c}y\colon \bigl(a\le y<b\land\varphi(y)\bigr)$. Along
this recursion, we collect in some set $V$ all the variables
$y_{\frac{1}{2}}$ seen in between that are required to satisfy
$\varphi$. At the very end of the recursion, we write down the formula
\[
  \forall y\colon\bigl((\bigvee_{x\in V} x=y) \to \varphi(y)\bigr)
\]
expressing all the ``postponed'' requirements at once.

The above idea is based on the linear order on the integers. If we
consider the non-unary quantifier $\exists^{=c}\bar{y}$, the role
of this linear order $\le$ is played by the lexicographic order on
tuples~$\bar{y}$.

The proof of the following lemma formalises the above ideas.  The
crucial requirement is that the formula and its block depth shall grow
only by a small summand (the latter makes sense only in case the
formula $\varphi$ does not contain any modulo-counting quantifiers,
i.e., belongs to $\FO$).

\begin{lem}\label{L-CMOD-to-FOMOD}
  Let $\alpha=\exists^{\ge c}\bar{y}\colon\varphi$ or
  $\alpha=\exists^{=c}\bar{y}\colon\varphi$ with
  $\varphi\in\FOMODmany$. There exists a formula $\psi\in\FOMODmany$
  with $\psi\iff\alpha$, $\CONST(\psi)=\CONST(\alpha)$,
  $\COEFF(\psi)=\COEFF(\alpha)$, and $\MOD(\psi)=\MOD(\alpha)$.

  Furthermore, $|\psi|\le|\varphi|+O(\ell\cdot\log c)$ where $\ell$ is
  the length of the tuple of variables $\bar{y}$ and the formula
  $\psi$ can be computed from $\alpha$ in time
  $|\varphi|+O(\ell\cdot\log c)$.

  If $\varphi$ belongs to $\FO$, then also $\psi\in\FO$ and the
  block depth of $\psi$ is at most
  $\bd^{\FO}(\varphi)+2\lceil\log(c)\rceil+2$.
\end{lem}

\begin{proof}
  Before formalising the above idea, we need some notational
  preparation. For an $\ell$-tuple of variables
  $\bar x=(x_1,\dots,x_\ell)$ and an assignment $f$, we write
  $f(\bar x)$ for the tuple
  $\bigl(f(x_1),f(x_2),\dots,f(x_\ell)\bigr)\in\bZ^\ell$. Furthermore
  (being a bit pedantic), we write $\exists\bar x$ for
  $\exists x_1\,\exists x_2\,\cdots\,\exists x_\ell$ and similarly for
  the universal quantifier.
  
  For $\ell\ge1$, let $\le^\ell_\lex$ denote the lexicographic order on
  $\bZ^\ell$. By induction on $\ell$, we construct formulas
  $(y_1,\dots,y_\ell)<^\ell_\lex(z_1,\dots,z_\ell)$ as follows:
  \begin{itemize}
  \item $y_1<^1_\lex z_1$ stands for $y_1<z_1$.
  \item $(y_1,\dots,y_\ell) <^\ell_\lex (z_1,\dots,z_\ell)$ stands for
    $y_1<z_1\lor
    y_1=z_1\land(y_2,\dots,y_\ell) <^{\ell-1}_\lex (z_2,\dots,z_\ell)
    $.
  \end{itemize}
  Then, for any assignment $f$, we have
  $f(\bar y <^\ell_\lex \bar z)=\true$ iff
  $f(\bar y) <^\ell_\lex f(\bar z)$. Similarly, the
  formulas
    \begin{align*}
      (\bar y \le^\ell_\lex \bar z) &= \lnot(\bar z <^\ell_\lex \bar y)
      \intertext{and}
      S(\bar y,\bar z)&=\bigl(\bar y<^\ell_\lex \bar z\land
                        \lnot\exists\bar x\colon
                         (\bar y <^\ell_\lex \bar x <^\ell_\lex\bar z)\bigr)
    \end{align*}
  hold under the assignment $f$ iff
  $f(\bar y) \le^\ell_\lex f(\bar z)$ and $f(\bar y)$ is the immediate
  predecessor of $f(\bar z)$ in $(\bZ^\ell,\le^\ell_\lex)$,
  respectively. Later, we will need that all these formulas $\beta$
  are of size $O(\ell)$ and satisfy
  $\CONST(\beta)=\COEFF(\beta)=\{0,\pm1,\pm2\}$ and
  $\MOD(\beta)=\{1\}$.

  We fix fresh $\ell$-tuples of variables $\links$, $\mitte$, $\rechts$,
  $\eins$, $\zwei$, and $\drei$ that have no
  variable in common.

  By induction on $n\ge0$, we will now construct for any finite set
  $V$ of $\ell$-tuples of variables a formula $\psi_{n,V}$ with the
  following property: Let $f$ be an assignment such that
  \begin{itemize}
  \item $f(\links)<^\ell_\lex f(\rechts)$ and
  \item no tuple $\bar v$ from $V$ satisfies
    $f(\links) \le^\ell_\lex f(\bar v) <^\ell_\lex f(\rechts)$.
  \end{itemize}
  In other words, the interval
  $\bigl[f(\links),f(\rechts)\bigr)\subseteq(\bZ^\ell,\le^\ell_\lex)$ is
  not empty, but contains none of the values $f(\bar v)$ for
  $\bar v\in V$.  Our construction of the formula $\psi_{n,V}$ will
  ensure that it holds under such an assignment~$f$, i.e.,
  $f(\psi_{n,V})=\true$, iff
  \begin{itemize}
  \item for all tuples $\bar v$ from $V$, we have
    $f\bigr(\varphi(\bar v)\bigr)=\true$ (more precisely:
    $f\bigl(\forall\bar x\colon (\bar x=\bar v\to\varphi)\bigr)=\true$)
    and
  \item there are precisely $n$ tuples $\bar m\in\bZ^\ell$ such that
    $f(\links) \le^\ell_\lex \bar m <^\ell_\lex f(\rechts)$ and
   $f_{\bar x/\bar m}(\varphi)=\true$.    
  \end{itemize}
  In this construction, it will be convenient to write $\bar w\in V$
  for $\bigvee_{\bar v\in V}\bar v=\bar w$, i.e., for the semantical
  property that $f(\bar w)$ is one of the tuples of integers
  $f(\bar v)$ with $\bar v\in V$.
  
  We start with $n=0$ and  $n=1$:
    \begin{align*}
      \psi_{0,V} &=\forall \bar x\colon
                   \bigl((\links \le^\ell_\lex \bar x <^\ell_\lex \rechts \lor \bar x\in V)\to
                   (\varphi\leftrightarrow\bar x\in V)\bigr)\\
      \psi_{1,V} &=\exists \mitte\colon
                   \begin{array}[t]{cl}
                     &\links \le^\ell_\lex \mitte <^\ell_\lex \rechts\\
                     \land
                     & \forall \bar x\colon\bigl((\links \le^\ell_\lex \bar x <^\ell_\lex \rechts \lor \bar x\in V)\to
                    (\varphi\leftrightarrow\bar x\in V\cup\{\mitte\})\bigr)
                   \end{array}
    \end{align*}

  For the induction step,
  we now construct $\psi_{2n,V}$ and $\psi_{2n+1,V}$ with $n\ge1$. The
  former is the simpler case:
    \[
      \psi_{2n,V}=\exists \eins,\zwei,\drei\colon
      \begin{array}[t]{cl}
        & \links=\eins <^\ell_\lex \zwei <^\ell_\lex \drei=\rechts\\
        \land
        &\forall \links,\rechts\colon
          \Bigl((\links,\rechts)\in\bigl\{(\eins,\zwei),(\zwei,\drei)\bigr\}
          \to \psi_{n,V}\Bigr)
      \end{array}
    \]
  Note that $(\links,\rechts)$ is a $2\ell$-tuple of variables so
  that $(\links,\rechts)\in\bigl\{(\eins,\zwei),(\zwei,\drei)\bigr\}$ is
  shorthand for the formula
    \[
      (\links=\eins\land \rechts=\zwei)\lor
      (\links=\zwei\land\rechts=\drei)\,.
    \]  
  The idea of the formula $\psi_{2n,V}$ is to divide the interval
  $\bigl[f(\links),f(\rechts)\bigr)=\bigl[f(\eins),f(\drei)\bigr)$
  into two subintervals $\bigl[f(\eins),f(\zwei)\bigr)$ and
  $\bigl[f(\zwei),f(\drei)\bigr)$ and to verify that both these
  intervals satisfy the formula~$\psi_{n,V}$, i.e., contain in particular
  precisely $n$ witnesses for $\varphi$.
  
  To also construct $\psi_{2n+1,V}$, we need another $\ell$-tuple
  $\zwei'$ of fresh variables and set
    \[
      \psi_{2n+1,V}=\exists \eins,\zwei',\zwei,\drei\colon\!\!\!\!\!\!\!\!
      \begin{array}[t]{cl}
        & \links=\eins <^\ell_\lex \zwei' <^\ell_\lex \zwei<^\ell_\lex \drei=\rechts\land S(\zwei',\zwei)\\
        \land
        &\forall \links,\rechts\colon
          \Bigl((\links,\rechts)\in\bigl\{(\eins,\zwei'),(\zwei,\drei)\bigr\}
          \to \psi_{n,V\cup\{\zwei'\}}\Bigr)\,.
      \end{array}
    \]
  Here, the idea is to divide the interval
  $I=\bigl[f(\links),f(\rechts)\bigr)=\bigl[f(\eins),f(\drei)\bigr)$
  into the half-open interval $I_1=\bigl[f(\eins),f(\zwei')\bigr)$ and
  the open interval $I_2=\bigl(f(\zwei'),f(\drei)\bigr)$ and to verify
  that both these intervals satisfy the
  formula~$\psi_{n,V\cup\{\zwei'\}}$, i.e., contain in particular
  precisely $n$ witnesses for~$\varphi$, and that $f(\zwei')$ satisfies
  $\varphi$. Since $I$ is the disjoint union of the intervals $I_1$,
  $\{f(\zwei')\}$, and $I_2$, this ensures that the interval $I$
  contains precisely $2n+1$ witnesses for $\varphi$.

  Then the formula
    \[
      \exists\links,\rechts\colon
          \bigl((\links <^\ell_\lex \rechts\land\psi_{c,\emptyset}\bigr)
    \]
  is equivalent to $\exists^{\ge c}\bar x\colon\varphi$ since it
  expresses that some interval contains precisely $c$ witnesses for
  $\varphi$.  Furthermore, the formula
    \[
      \exists\eins,\zwei\, \forall \links,\rechts\colon
          \bigl((\links\le^\ell_\lex \eins\land \zwei\le^\ell_\lex \rechts)
          \to \psi_{c,\emptyset}\bigr)
    \]
  is equivalent to $\exists^{=c}\bar x\colon\varphi$ since it
  expresses that for some interval, any superinterval contains
  precisely $c$ witnesses for $\varphi$.

  It remains to analyse the size of the resulting formula as well as
  the block depth in case $\varphi\in\FO$.
  
  To estimate the size of $\psi_{c,\emptyset}$, note the following:
  \begin{itemize}
  \item The size of the formulas $\psi_{0,V}$ and $\psi_{1,V}$ is of
    the form $|\varphi|+O\bigl(\ell\cdot\log(c)\bigr)$ since we allow
    the Boolean connective $\leftrightarrow$ in our formulas and since
    the size of $V$ is bounded by $\lceil\log(c)\rceil$ (the formula
    size doubles if we consider $\leftrightarrow$ as abbreviation).
  \item The size increase when moving from $\psi_{n,V}$ to
    $\psi_{2n,V}$ is bounded by a summand of size $O(\ell)$ and the
    same applies to the construction of $\psi_{2n+1,V}$ from
    $\psi_{n,V\cup\{\zwei'\}}$.
  \end{itemize}
  It follows that
  $|\psi_{c,\emptyset}|\le |\varphi|+\varkappa\cdot\ell\cdot\log(c)$
  for some constant $\varkappa$. One sees easily that the same holds
  for the formula $\psi$ and that it can be constructed in time
  $|\varphi|+O\bigl(\ell\cdot\log(c)\bigr)$.

  Now suppose $\varphi\in\FO$. Since in the construction, we only
  introduce classical existential quantifiers, we obtain
  $\psi_{c,\emptyset}\in\FO$.  We want to analyse the block depth of
  $\psi_{c,\emptyset}$. First note that 
    \begin{align*}
      \bd^{\FO}(\psi_{0,V}),\bd^{\FO}(\psi_{1,V}) &\le \bd^{\FO}(\varphi)+2\,,\\
      \bd^{\FO}(\psi_{2n,V}) &\le \bd^{\FO}(\psi_{n,V})+2\,,\text{ and}\\
      \bd^{\FO}(\psi_{2n+1,V}) &\le \bd^{\FO}(\psi_{n,V\cup\{\zwei'\}})+2\,.
    \end{align*}
    It follows that
    $\bd^{\FO}(\psi_{c,V})\le\bd^{\FO}(\varphi)+2\cdot\lceil\log(c)\rceil$.  In
    the final step, the block depth increases by at most $2$. Hence we
    obtain $\bd^{\FO}(\psi)\le\bd^{\FO}(\varphi)+2\lceil\log(c)\rceil+2$.
\end{proof}

The above lemma can be applied iteratively to all threshold- and
exact-counting quantifiers. Hence, from a formula from $\fullLogic$,
we obtain an equivalent formula in $\FOMODmany$, and from a formula
from $\C$, we obtain a formula from $\FO$. In order to bound the block
depth of this formula from $\FO$, we extend its definition to formulas
from $\C$ as follows:
\begin{defi}\hfill
  \begin{itemize}
  \item $\BD_0$ is the set of atomic formulas.
  \item For $n\ge1$, the set $\BD_n$ contains the formulas of the
    following forms:
    \begin{itemize}
    \item $\exists x_1\,\exists x_2\,\dots\,\exists x_m\colon\beta$
      where $m\ge0$ and $\beta$ is a Boolean combination (possibly
      using $\lnot$, $\land$, $\lor$, $\to$, and $\leftrightarrow$) of
      formulas from $\BD_{n-1}$
    \item $\exists^{\ge c}\overline{x}\colon\beta$ or
      $\exists^{=c}\overline{x}\colon\beta$ where $\beta$ is a Boolean
      combination of formulas from $\BD_{n-2\lceil\log_2 c\rceil-2}$
  \end{itemize}
  \item The \emph{block depth} $\bd(\alpha)$ of a formula $\alpha \in \C$ is
    the minimal natural number $n$ with $\alpha\in\BD_n$.
  \end{itemize}
\end{defi}

Note that the block depth of a formula from 
$\C$ is at most twice the length of the formula (since the constants
$c$ in $\exists^{\ge c}$ and $\exists^{=c}$ are written in
binary). Furthermore, if $\alpha\in\FO$, then
$\bd^{\FO}(\alpha)=\bd(\alpha)$.

\begin{prop}\label{P-CMOD-to-FOMOD}
  From a formula $\varphi\in\fullLogic$, one can construct in time
  polynomial in $|\varphi|$ an equivalent formula
  $\psi\in\FOMODmany$.

  In addition, we have $\COEFF(\psi)\subseteq\COEFF(\varphi)$,
  $\CONST(\psi)\subseteq\CONST(\varphi)$, and
  $\MOD(\psi)\subseteq\MOD(\varphi)$.

  Furthermore, if $\varphi\in\C$, then the resulting formula $\psi$
  belongs to $\FO$ and the block depth $\bd(\psi)$ of $\psi$ equals
  that of $\varphi$.
\end{prop}

\begin{proof}
  Let $\varphi_0=\varphi$ contain $n$ threshold- or exact-counting
  quantifiers. We construct, inductively, formulas $\varphi_{i+1}$
  from $\varphi_i$ using Lemma~\ref{L-CMOD-to-FOMOD} that contain one
  threshold- or exact-counting quantifier less. When constructing
  $\varphi_{i+1}$, suppose we eliminate a quantifier of the form
  $\exists^{\ge c_i}(y_1,\dots,y_{\ell_i})$ or
  $\exists^{=c_i}(y_1,\dots,y_{\ell_i})$. Then $\varphi_{i+1}$ can be
  constructed from $\varphi_i$ in time
  $|\varphi_i|+O\bigl(\ell_i\cdot\log(c_i)\bigr)$. Since
  $\sum_{0\le i<n}\ell_i \le |\varphi|$ and
  $\sum_{0\le i<n}\log(c_i) \le |\varphi|$, the construction of
  $\varphi_n$ can be carried out in time polynomial in $|\varphi|$.

  Now suppose $\varphi\in\C$. Then the formula $\varphi_n$ belongs to
  $\FO$. Furthermore, when moving from $\varphi_i$ to $\varphi_{i+1}$,
  the block depth does not increase.
\end{proof}

From Berman's upper bound for Presburger arithmetic, we get
immediately the following for the logic $\C$, i.e., the fragment of
$\fullLogic$ without modulo-counting quantifiers.
\begin{cor}\label{C-complexity-of-C}
  Satisfaction of a closed formula $\varphi\in\C$ can be decided in
  doubly exponential alternating time with linearly many alternations.
\end{cor}
\begin{proof}
  The transformation of $\varphi$ into an equivalent closed formula
  $\psi$ from $\FO$ increases the size of the formula only
  polynomially and the resulting block depth belongs to
  $O(|\varphi|)$. Hence the claim follows from Berman's
  Theorem~\ref{T-Berman}.
\end{proof}

Somewhat surprisingly, the above result says that adding the
quantifiers $\exists^{\ge c}\bar{x}$ and $\exists^{=c}\bar{x}$ does
not increase the complexity of the decision procedure; for the unary
version of the above logic, i.e., for $\Cone$, this was already
observed in \cite{ChiHM22}.

\subsection{Elimination of non-unary modulo-counting quantifiers}

Here, we give the transformation from $\FOMODmany$ to $\FOMOD$.
We will provide a transformation that can be computed in doubly
exponential time and does not change the sets $\COEFF$, $\CONST$, and
$\MOD$, nor the quantifier depth.

The crucial task in this section is to express a non-unary
quantification $\exists^{(t,p)}(y_1,\dots,y_\ell)$ (where the
remainder is given as a term $t$) using only unary modulo-counting
quantifications where the remainder is given as a constant. The first
step is obvious: Using a case distinction, we replace
$\exists^{(t,p)}\bar{y}\colon\varphi$ by the disjunction of all
formulas $t\equiv_p r\land\exists^{(r,p)}\bar{y}\colon\varphi$ for
$0\le r<p$. As a second step, one has to eliminate the quantification
over a tuple $\bar{y}$.  We explain the basic idea using the formula
$\exists^{(0,2)}(y_1,y_2)\colon \rho(y_1,y_2)$ where $\rho$ is some
formula and $R$ is the set of pairs of integers satisfying $\rho$. We
have to express that $R$ is finite and its number of elements is even.

Assuming $R$ to be finite, its size is even iff the number of elements
$y_1$ with
\[
  \Bigl|\bigl\{y_2\bigm| (y_1,y_2)\in R\bigr\}\Bigr|\text{ odd}
\]
is even. This can be expressed by the formula
\[
  \exists^{(0,2)}y_1\,\exists^{(1,2)}y_2\colon \rho(y_1,y_2)\,.
\]

Further, $R$ is finite iff its number of elements is even or odd. But
this would not eliminate the non-unary quantification. Alternatively,
$R$ is finite iff it is bounded, i.e., if
\[
  \exists z\,\forall y_1\,\forall y_2\,\bigl(\rho(y_1,y_2)\to |y_1|,|y_2|\le z\bigr)
\]
holds. Although being a simple formula, its quantifier rank is larger
than that of the formula we started with.  Yet another
characterisation of finiteness of $R$ is ``only finitely many elements
can be extended to a tuple from $R$ and no element can be extended in
infinitely many ways''. The following formula expresses precisely
this:
\begin{align*}
  & \exists^{(0,2)}y_1\,\exists y_2\colon \rho(y_1,y_2)
    \lor
    \exists^{(1,2)}y_1\,\exists y_2\colon \rho(y_1,y_2)\\
  \land
  & \forall y_1\,
    \Bigl(
    \exists^{(0,2)}y_2\colon \rho(y_1,y_2)
    \lor
    \exists^{(1,2)}y_2\colon \rho(y_1,y_2)
    \Bigr)
\end{align*}
The proof of the following lemma formalises this idea (and extends it
to other moduli and remainder given as terms). In other words, it
shows how to eliminate a single non-unary modulo-counting quantifier.

\begin{lem}\label{L-FOMOD-to-FOMOD-1}
  Let $\alpha=\exists^{(t,p)}\bar y\colon \varphi$ with
  $\varphi\in\FOMOD$. There exists a formula $\psi\in\FOMOD$
  with $\psi\iff\alpha$, $\CONST(\psi)=\CONST(\alpha)$,
  $\COEFF(\psi)=\COEFF(\alpha)$, $\MOD(\psi)=\MOD(\alpha)$, and
  $\qd(\psi)=\qd(\alpha)$.

  Furthermore, $\psi$ can be constructed from $\alpha$ in time
  $O\bigl(p^{p\cdot\ell}\cdot|\alpha|\bigr)$
  where $\ell$ is the length
  of the tuple $\bar y$.
\end{lem}

Note that the modulus $p$ is given in binary. Hence the time bound
is doubly exponential in the size of the formula
$\alpha$.

\begin{proof}
  First suppose $\ell=1$ and consider the formula 
  \[
    \psi:=\bigvee_{0\le r<p}
       (r\equiv_p t\land\exists^{(r,p)}y_1\colon\varphi)
  \]
  which is clearly equivalent to $\alpha$ and has all the properties
  required by the claim of the lemma.

  So suppose $\ell>1$. First, we construct inductively a formula from
  $\FOMOD$ expressing that there are only finitely many tuples
  $\overline{y}$ satisfying $\varphi$ (for $0\le n<\ell-1$):
  \begin{align*}
    \eta_{\ell-1}(y_1,\dots,y_{\ell-1}) = 
    &\bigvee_{0\le i<p}
       \exists^{(i,p)}y_\ell\colon \varphi\\
    \eta_n(y_1,\dots,y_n) = 
    &\bigvee_{0\le i<p}
       \exists^{(i,p)}y_{n+1}\,\exists(y_{n+2},\dots,y_\ell)\colon \varphi\\
    &\land \forall y_{n+1}\colon\eta_{n+1}(y_1,\dots,y_{n+1})
  \end{align*}
  Let $(y_1,\dots,y_{\ell-1})$ be any tuple of integers. The formula
  $\eta_{\ell-1}$ expresses that the tuple $(y_1,\dots,y_{\ell-1})$
  can be extended to a tuple satisfying $\varphi$ in only finitely
  many ways.

  Now let $(y_1,\dots,y_{\ell-2})$ be any tuple of integers. The first
  line of the formula $\eta_{\ell-2}$ expresses that the tuple
  $(y_1,\dots,y_{\ell-2})$ can be extended to a tuple
  $(y_1,\dots,y_{\ell-1})$ that allows a further extension to a tuple
  satisfying $\varphi$ in only finitely many ways.  The second line
  expresses that, for any integer $y_{\ell-1}$, there are only
  finitely many extensions of the tuple $(y_1,\dots, y_{\ell-1})$ to a
  tuple $(y_1,\dots,y_\ell)$ satisfying $\varphi$. Hence,
  $\eta_{\ell-2}$ expresses that there are only finitely many
  extensions of $(y_1,\dots,y_{\ell-2})$ satisfying $\varphi$.

  Arguing inductively, we obtain that $\eta_0$ expresses that there
  are only finitely many tuples $(y_1,\dots,y_\ell)$ satisfying
  $\varphi$.

  Now consider the formula
  \[
    \beta = \bigvee_{0\le r<p}\bigl(r\equiv_p t\land\eta_0\land\exists^{(r,p)}\overline{y}\colon\varphi\bigr)
  \]
  that is equivalent with $\alpha$. It remains to rewrite
  $\exists^{(r,p)}\overline{y}\colon\varphi$ into a formula from
  $\FOMOD$. In this construction, we can assume that $\eta_0$
  holds, i.e., that there are only finitely many tuples
  $(y_1,\dots,y_\ell)$ satisfying $\varphi$. To this aim, consider the
   $\FOMOD$-formulas (for $0\le n<\ell-1$ and $0\le d<p$)
  \begin{align*}
    \delta_{\ell-1}^d(y_1,\dots,y_{\ell-1})  &=
        \exists^{(d,p)}y_\ell\colon \varphi\text{ and}\\
    \delta_n^d(y_1,\dots,y_n)&=
    \bigvee_{(*)}
      \bigwedge_{0< i<p}
        \exists^{(d_i,p)}y_{n+1}\colon
          \delta_{n+1}^i(y_1,\dots,y_{n+1})
  \end{align*}
  where the disjunction $(*)$ extends over all tuples
  $(d_1,\dots,d_{p-1})$ over $\{0,1,\dots,p-1\}$ such that
  \begin{equation}
    \label{eq:sum}
    \sum_{0< i<p}d_i\cdot i\equiv_p d\,.
  \end{equation}
  Let $(y_1,\dots,y_{\ell-1})$ be a tuple of integers. Then the
  formula $\delta_{\ell-1}^d$ expresses that there are $d$ many ways
  to extend the tuple $(y_1,\dots,y_{\ell-1})$ to a tuple satisfying
  $\varphi$ (all counts in this paragraph are understood modulo $p$).
  Next let $(y_1,\dots,y_{\ell-2})$ be a tuple of integers. Then the
  conjunction in the formula $\delta_{\ell-2}^d$ expresses that, for
  all $i\in\{1,2,\dots,p-1\}$, there are $d_i$ many values for
  $y_{\ell-1}$ that satisfy
  $\delta^i_{\ell-1}(y_1,\dots,y_{\ell-1})$, i.e., that can be
  extended in $i$ many ways to a tuple satisfying $\varphi$. Thus, the
  formula $\delta_{\ell-2}^d$ expresses that the tuple
  $(y_1,\dots,y_{\ell-2})$ can be extended in $d$ many ways to a tuple
  satisfying $\varphi$.
  Arguing inductively, the formula $\delta_0^d$ expresses that there
  are $d$ many tuples satisfying $\varphi$.

  Setting
  \[
    \psi:=\eta_0\land\bigvee_{0\le r<p}\bigl(r\equiv_p t\land\delta_0^r\bigr)
  \]
  we consequently get $\alpha\Leftrightarrow\beta\Leftrightarrow\psi\in\FOMOD$.

  Note that the construction of $\eta_0$ and $\delta_0$ leaves the
  sets $\COEFF(.)$, $\MOD(.)$, and $\CONST(.)$ and the
  quantifier-depth unchanged.

  It remains to bound the time needed to construct the formula $\psi$.
  First, $\eta_0$ can be constructed in time $O(\ell\cdot p\cdot|\alpha|)$ since
  the formula $\eta_{n+1}$ appears only once in $\eta_n$. Next, any of
  the formulas $\delta_{\ell-1}^d$ can be constructed in time
  $O(|\alpha|)$.  We now consider the construction of $\delta_n^d$
  from the formulas $\delta_{n+1}^i$. Note that the tuple
  $(d_1,\dots,d_{p-2})$ together with equation~\eqref{eq:sum}
  completely determines the value of
  $d_{p-1}\in\{0,\dots,p-1\}$. Hence the disjunction $(*)$ extends
  over at most $p^{p-2}$ tuples. Consequently, the formula
  $\delta_n^d$ contains at most $p^{p-2}\cdot(p-1)\le p^{p-1}$ many
  subformulas $\delta_{n+1}^i$. By induction, we obtain that
  $\delta_0^r$ can be constructed in time
  $O\bigl(p^{(p-1)\cdot\ell}\cdot|\alpha|\bigr)$. Since the
  construction of $\psi$ requires this to be done for all
  $r\in\{0,1,\dots,p-1\}$ and furthermore
  $r\equiv_pt$ has to be added, the formula $\psi$ can be constructed in time 
  $O\bigl(p\cdot \log(p) \cdot p^{(p-1)\cdot\ell}\cdot|\alpha|\bigr)$
  which is in 
  $O\bigl(p^{p\cdot\ell}\cdot|\alpha|\bigr)$ as $\ell > 1$.
\end{proof}

The above lemma allows to reduce the number of non-unary
modulo-counting quantifiers by one, hence an inductive application
eliminates all of them. The algorithmic cost and the form of the
resulting formula is analysed in the following proof.

\begin{prop}\label{P-FOMOD-to-FOMOD-}
  From a formula $\varphi\in\FOMODmany$, one can construct in time
  doubly exponential in $|\varphi|$ an equivalent formula
  $\gamma\in\FOMOD$.

  In addition, we have $\COEFF(\gamma)\subseteq\COEFF(\varphi)$,
  $\CONST(\gamma)\subseteq\CONST(\varphi)$,
  $\MOD(\gamma)\subseteq\MOD(\varphi)$, and $\qd(\gamma)\le\qd(\varphi)$.
\end{prop}

\begin{proof}
  Let $P$ be the maximal value such that some modulo-counting
  quantifier $\exists^{(t,P)}$ appears in the formula~$\varphi$ and
  let $L$ be the maximal arity of any modulo-counting quantifier in
  $\varphi$. Finally, let $n$ be the number of non-unary
  modulo-counting quantifiers in $\varphi$.

  Let $\varphi_0=\varphi$. To inductively construct $\varphi_{i+1}$
  from $\varphi_i$, we chose some subformula
  $\exists^{(t,p)}(y_1,\dots,y_\ell)\colon\alpha$ with $\ell>1$ and
  $\alpha\in\FOMOD$. This subformula is replaced by an equivalent
  formula from $\FOMOD$ that we obtain from
  Lemma~\ref{L-FOMOD-to-FOMOD-1}. This reduces the number of non-unary
  modulo-counting quantifiers by one so that $\gamma:=\varphi_n$ is
  a formula from $\FOMOD$.

  From Lemma~\ref{L-FOMOD-to-FOMOD-1}, we get $\gamma\iff\varphi$,
  $\CONST(\gamma)=\CONST(\varphi)$, $\COEFF(\gamma)=\COEFF(\varphi)$,
  $\MOD(\gamma)=\MOD(\varphi)$, and $\qd(\gamma)=\qd(\varphi)$.

  Also from Lemma~\ref{L-FOMOD-to-FOMOD-1}, we get that
  $\varphi_{i+1}$ can be constructed from $\varphi_i$ in time
  $O(P^{P\cdot L}\cdot|\varphi_i|)$ and is therefore of size at most
  $O(P^{P\cdot L}\cdot|\varphi_i|)$. Consequently, $\gamma$ can be
  constructed from $\varphi_0$ in time
  $O(\bigl(P^{P\cdot L}\bigr)^n\cdot|\varphi|)$. Since the binary
  encoding of $P$ appears in $\varphi$, we get
  $P\le2^{|\varphi|}$. Furthermore, $L,n\le|\varphi|$. Consequently, the
  construction of $\gamma$ from $\varphi$ can be carried out in doubly
  exponential time.
\end{proof}

The above two Propositions~\ref{P-CMOD-to-FOMOD} and
\ref{P-FOMOD-to-FOMOD-} imply the following.
\begin{thm}\label{P-CMOD-to-FOMOD-}
  From a formula $\varphi\in\fullLogic$, one can construct in time
  doubly exponential in $|\varphi|$ an equivalent formula
  $\gamma\in\FOMOD$.

  In addition, we have $\COEFF(\gamma)\subseteq\COEFF(\varphi)$,
  $\CONST(\gamma)\subseteq\CONST(\varphi)$, and
  $\MOD(\gamma)\subseteq\MOD(\varphi)$.

  In addition, the quantifier depth $\qd(\gamma)$ is polynomial in the
  size of $\varphi$.
\end{thm}

\begin{proof}
  Using Proposition~\ref{P-CMOD-to-FOMOD}, one first constructs in
  polynomial time an equivalent formula $\psi$ from $\FOMODmany$. This
  formula is then, using Proposition~\ref{P-FOMOD-to-FOMOD-},
  translated into an equivalent formula $\gamma$ from $\FOMOD$.

  Since $|\psi|$ is polynomial in the size of $\varphi$, its
  quantifier depth is also polynomial in $|\varphi|$. Hence, the same
  holds for the quantifier depth of $\gamma$.
\end{proof}
\section{Quantifier elimination}

This section provides a quantifier elimination procedure for the logic $\FOMOD$
where, differently from the full logic $\fullLogic$, only unary
quantifications $\exists y$ and $\exists^{(q,p)}y$ with $q\in\bN$ are
allowed.

As usual with quantifier elimination procedures, we first demonstrate
how to eliminate a single quantifier in front of a Boolean combination
of atomic formulas. Since the classical existential quantifier and the
modulo-counting quantifier behave rather differently, we handle them
in separate Lemmas~\ref{LB} and \ref{LB'}. The main point in both
these lemmas is
\begin{enumerate}[(a)]
\item properties of the form $\exists/\exists^{(q,p)}x\colon \beta$
  where $\beta$ is quantifier-free can be expressed without
  quantification and
\item the sets of coefficients, constants, and moduli vary in this
  process, but these sets can be controlled.
\end{enumerate}
Our quantifier elimination is effective, but we do not concentrate on
this fact. We do, in particular, not aim at a fast elimination
algorithm nor at small resulting formulas. All we need for our later
decision procedure is a bound on the size of the coefficients,
constants, and moduli appearing in the resulting formula.

For this bound, suppose $\beta$ is a quantifier-free formula and $E$
is a quantifier $\exists$ or $\exists^{(q,p)}$. We will prove that
$Ex\colon\beta$ is equivalent to some quantifier-free formula $\gamma$
whose sets of coefficients etc.\ are contained in the following sets
(with $p=1$ in case $E=\exists$):
\begin{align*}
  \COEFF_p(\beta) & =
      \bigl\{a_1a_2-a_3a_4\bigm| a_1,a_2,a_3,a_4\in\COEFF(\beta)\bigr\}\\
  \CONST_p(\beta) & =
     \left\{ a_1c_1-a_2(c_2+c)
         \,\left|\,
           \begin{array}{l}
             a_1,a_2\in\COEFF(\beta), c_1,c_2\in\CONST(\beta)\\
             |c|\le\max\COEFF(\beta)\cdot p\cdot\lcm\MOD(\beta)
           \end{array}
         \right\}\right.\\
    \MOD_p(\beta) & =
      \bigl\{a_1a_2k_1k_2\bigm| a_1a_2\in\COEFF(\beta),k_1,k_2\in\MOD(\beta)\cup\{p\}\bigr\}
\end{align*}
Note that the first set does not depend on the number $p$ and that
$\CONST_p(\beta)\subseteq\CONST_{p_1}(\beta)$ for all $1\le p<p_1$.

Using these sets, we formulate the following condition on the triple
$(\beta,\gamma, p)$ where $\beta$ and $\gamma$ are formulas and
$p\ge1$ is a positive integer:
\begin{equation}
  \label{eq:bedingung}
    \COEFF(\gamma) \subseteq \COEFF_p(\beta)\,,\  
    \CONST(\gamma) \subseteq \CONST_p(\beta)\,,
    \MOD(\gamma) \subseteq \MOD_p(\beta)
\end{equation}

Let $\beta$ be a quantifier-free formula and $x=t$ an equation (with
$t$ an $x$-free term). Write $\beta'$ for the formula obtained from
$\beta$ by replacing all occurrences of $x$ by $t$ so that $\beta'$
is a Boolean combination of $x$-free atomic formulas. Then the
formulas $x=t\land\beta$ and $x=t\land\beta'$ are equivalent. The
following lemma, whose statement will be used repeatedly, demonstrates
the analogous fact for equations of the form $ax=t$ (with $a\neq0$),
i.e., constructs an $x$-free quantifier-free formula $\beta'$ so
that $ax=t\land\beta$ and $ax=t\land\beta'$ are equivalent. The main
point here is that, under a specific condition on $a$, $t$, and $c$,
the triple $(\beta,\beta',p)$ satisfies the above
Condition~\eqref{eq:bedingung}.

\begin{lem}\label{L-neu}
  Let $\beta$ be a Boolean combination of $x$-separated atomic
  formulas, $ax < t$ or $t < ax$ some atomic formula from $\beta$ with
  $a> 0$, $p\ge1$ a positive integer, and $c\in\bZ$ with
  $|c|\le a\cdot p\cdot \lcm\MOD(\beta)$. There exists a Boolean
  combination $\beta_{a,t+c}$ of $x$-free atomic formulas such that
  the triple $(\beta,\beta_{a,t+c},p)$ satisfies
  Condition~\eqref{eq:bedingung} and, for all assignments $f$,
    \[
      f(ax)=f(t+c) \text{ implies } f(\beta)=f(\beta_{a,t+c})\,.
    \]
\end{lem}

Note that in particular
  \[
     ax=t+c\land\beta \Longleftrightarrow ax=t+c\land\beta_{a,t+c}\,.
  \]
\begin{proof}
 The formula $\beta_{a,t+c}$ is obtained from $\beta$ by the following
  replacements (where $s$ is some $x$-free term, $a'\ge0$, and $k\ge2$):
  \begin{center}
    \begin{tabular}{rcl}
      $a'x < s$ & is replaced by & $a't+a'c < as$\\
      $s < a'x$ & is replaced by & $as < a't+a'c$\\
      $a'x \equiv_k s$ & is replaced by & $a't + a'c\equiv_{ak} as$
    \end{tabular}
  \end{center}

  Let $f$ be some assignment with $f(ax)=f(t+c)$. Then we have
    \begin{align*}
      f(a'x<s) & = f(a'ax<as) &&\text{since }a>0\\
               & = f\bigl(a'(t+c)<as\bigr)&&\text{since }f(ax)=f(t+c)\\
               & = f(a't+a'c < as)
    \intertext{and similarly}
      f(s<a'x) & = f(as < a't+a'c)
    \intertext{as well as}
      f(a'x\equiv_k s) & = f(a'ax \equiv_{ak} as)\\
                       & = f\bigl(a'(t+c)\equiv_{ak}as\bigr)\\
                       & = f(a't+a'c\equiv_{ak} as)\,.
    \end{align*}
  This completes the proof that $f(ax)=f(t+c)$ implies
  $f(\beta)=f(\beta_{a,t+c})$.

  It remains to verify Condition~\eqref{eq:bedingung}. First note that
  $a\in\COEFF(\beta)$ since $ax<t$ or $ax>t$ appears in $\beta$ and
  since $t$ is $x$-free.

  Now, let $b\in\COEFF(\beta_{a,t+c)}$. If $b\in\COEFF(\beta)$, we get
  $b=1b-0b$ implying $b\in\COEFF_p(\beta)$ since
  $1,0\in\COEFF(\beta)$. So let $b\notin\COEFF(\beta)$. Then there
  exists some atomic formula $a'x<s$ or $s<a'x$ in $\beta$ such that
  $b$ is some coefficient in the term $as-a'(t+c)$. Consequently,
  there exists a variable $y$ with coefficient $a_2$ in $s$ and with
  coefficient $a_4$ in $t$ such that $b=aa_2-a'a_4$. Since $a'x<s$ or
  $s<a'x$ is an atomic formula in $\beta$ and since $s$ is $x$-free,
  we have $a'\in\COEFF(\beta)$. Hence, also in this case,
  $b\in\COEFF_p(\beta)$.

  Next let $d\in\CONST(\beta_{a,t+c})$. If $d\in\CONST(\beta)$, we
  have $d=1d-0(0+c)\in\CONST_p(\beta)$. So suppose
  $d\notin\CONST(\beta)$. Then, as above, there exists some atomic
  formula $a'x<s$ or $s<a'x$ in $\beta$ such that $\pm d$ is the
  constant term in $as-a'(t+c)$. Consequently, $\pm d=ac_1-a'(c_2+c)$
  where $c_1$ and $c_2$ are the constant terms of $s$ and $t$,
  resp. Since $a,a'\in\COEFF(\beta)$ (see above) and since
  $|c|\le a\cdot p\cdot \lcm\MOD(\beta)$, we get $d\in\CONST_p(\beta)$.

  Finally, let $\ell\in\MOD(\beta_{a,t+c})$. If $\ell\in\MOD(\beta)$,
  then $\ell=1\cdot1\cdot\ell\cdot 1\in \MOD_p(\beta)$ since
  $1\in\COEFF(\beta)$. Otherwise, there exists an atomic formula
  $a'x\equiv_k s$ in $\beta$ with $\ell=ak$. Hence, also in this case,
  $\ell=1\cdot ak\cdot 1\in\MOD_p(\beta)$.
\end{proof}

We now come to the elimination of the classical existential
quantifier. Neither the result nor its proof are new, we present them
here to be able to also verify Condition~\eqref{eq:bedingung}.

\begin{lem}\label{LB}
  Let $x$ be a variable and $\beta$ a Boolean combination of
  $x$-separated atomic formulas. Then there exists a Boolean
  combination $\gamma$ of $x$-free atomic formulas such that the
  triple $(\beta,\gamma,1)$ satisfies Condition~\eqref{eq:bedingung}
  and $(\exists x\colon\beta)\Longleftrightarrow\gamma$.
\end{lem}

\begin{proof}
  Let $T$ be the set of all pairs $(a,t)$ such that $\beta$ contains
  an atomic formula of the form $ax<t$ or $t<ax$ with $a>0$. We first
  assume that this set $T$ is not empty. Let furthermore
  $N=\lcm\bigl(\MOD(\beta)\bigr)$. In particular, $N$ is a multiple of
  every integer~$k$ such that the atomic formula $ax\equiv_k t$
  appears in $\beta$ for some term $t$ and some $a\in\bZ$. Then we set
  \[
     \gamma :=  \bigvee (\beta_{a,t+c}\land 0\equiv_a t+c)
  \]
  where the disjunction extends over all triples $(a,t,c)$ with
  $(a,t)\in T$ and $-aN\le c\le aN$ (since $T\neq\emptyset$, this
  disjunction is not empty). We prove
  $(\exists x\colon \beta)\Longleftrightarrow\gamma$. So let $f$ be an
  assignment with $f(\exists x\colon\beta)=\true$. Then there is
  $b\in\bZ$ with $f_{x/b}(\beta)=\true$. Let $g=f_{x/b}$. Since the
  values $\frac{f(t)}{a}$ for $(a,t)\in T$ divide $\bZ$ into
  intervals, there exists $(a,t)\in T$ such that
  \begin{enumerate}
  \item $b=\frac{f(t)}{a}$ or
  \item $\frac{f(t)}{a}<b$ and for all $(a',t')\in T$ with
    $\frac{f(t')}{a'}<b$, we have
    $\frac{f(t')}{a'} \le \frac{f(t)}{a}$ or
  \item $b < \frac{f(t)}{a}$ and for all $(a',t')\in T$ with
    $b < \frac{f(t')}{a'}$, we have $\frac{f(t)}{a} \le \frac{f(t')}{a'}$.
  \end{enumerate}
  (The 2nd and 3rd cases are not exclusive, but if $b<\frac{f(t)}{a}$
  for all $(a,t)\in T$, then only the third case applies and
  symmetrically in case $b>\frac{f(t)}{a}$.)  Assume the first case. Then
  $g(ax)=ab=f(t)=g(t)$ where the last equality holds since $t$ is
  $x$-free. Hence, by Lemma~\ref{L-neu}, we get
  $\true=g(\beta)=g(\beta_{a,t})=f(\beta_{a,t})$ and, since
  $\frac{f(t)}{a}=b\in\bZ$, also $f(0\equiv_a t)=\true$. Hence, using
  the triple $(a,t,0)$, we have $f(\gamma)=\true$.

  Next consider the second case. There exists $k\in\bN$ with
  $0<(b-kN)-\frac{f(t)}{a} \le N$ or, equivalently, $0 < a(b-kN)-f(t)
  \le aN$. We set $c=a(b-kN)-f(t)$ so that $-aN \le c \le aN$.

  Since $N$ is a multiple of all moduli appearing in $\beta$, we get
  $f_{x/b-kN}(\beta)=\true$ from $f_{x/b}(\beta)=\true$ and the choice
  of $(a,t)$ and of $k$. Set $g'=f_{x/b-kN}$. Then
  $g'(ax) = a(b-kN) = f(t+c) = g'(t+c)$ since the term $t+c$ is
  $x$-free. Hence, by Lemma~\ref{L-neu}, we get
  $\true = f_{x/b-kN}(\beta) = g'(\beta) = g'(\beta_{a,t+c}) =
  f(\beta_{a,t+c})$. Furthermore, $f(t)+c = a(b-kN)$ is divisible by
  $a$ so that $f(0\equiv_a t+c)=\true$. Using the triple $(a,t,c)$,
  we obtain $f(\gamma)=\true$ also in the second case.

  The third case is symmetric to the second, i.e., we showed
  $f(\exists x\colon\beta)=\true \Longrightarrow f(\gamma)=\true$.

  For the converse implication, suppose $f(\gamma)=\true$. Then there is a
  triple $(a,t,c)$ with $(a,t)\in T$ and $-aN \le c \le aN$ such that
  $f(\beta_{a,t+c}\land 0\equiv_a t+c)=\true$. Because of
  $0\equiv_a f(t)+c$, there exists $b\in\bZ$ with $ab= f(t+c)$. Let
  $g=f_{x/b}$. Then $g(ax)=ab=f(t+c)=g(t+c)$ since $t$ is
  $x$-free. Hence, by Lemma~\ref{L-neu}, we have
  $g(\beta)=g(\beta_{a,t+c})=f(\beta_{a,t+c})=\true$. Since
  $g=f_{x/b}$, this implies $f(\exists x\colon\beta)=\true$ and
  therefore the remaining implication.

  Finally, we have to verify Condition~\eqref{eq:bedingung}. Recall
  that $(a,t)\in T$ means that $ax<t$ or $t<ax$ is a subformula of
  $\beta$ (or $a=1$ and $t=0$). Hence
  $\COEFF(\gamma)\subseteq\COEFF_1(\beta)$ and
  $\CONST(\gamma)\subseteq\CONST_1(\beta)$ follow immediately from
  Lemma~\ref{L-neu} since these sets only refer to atomic formulas of
  the form $a'x<s$ or $a'x>s$. Next let $\ell\in\MOD(\gamma)$. Then
  $\ell\in\MOD(\beta_{a,t+c})$ or $\ell=a$ for some $(a,t)\in T$ and
  $|c|\le aN$. In the first case, $\ell\in\MOD_1(\beta)$ follows from
  Lemma~\ref{L-neu}, in the latter case note that
  $a,1\in\COEFF(\beta)$ and $1\in\MOD(\beta)$ so that
  $\ell=a=1\cdot a\cdot 1\cdot 1\in\MOD_1(\beta)$.

  Thus, we proved the lemma in case $T\neq\emptyset$. Now assume
  $T=\emptyset$. Note that the formulas $\beta$ and
  $\beta\land(x<0\lor\lnot x<0)$ are equivalent, agree on the sets of
  coefficients etc., and that the latter contains some atomic formula
  of the form $ax<t$. Thus, by the above arguments, we find the
  Boolean combination $\gamma$ with the desired properties also in
  this case.
\end{proof}

Having shown how to eliminate a single existential quantifier, we now
come to the analogous result for modulo-counting quantifiers.

\begin{lem}\label{LB'}
  Let $x$ be a variable, $\beta$ a Boolean combination of
  $x$-separated atomic formulas, and $0\le q<p$ natural numbers. Then
  there exists a Boolean combination of $x$-free atomic formulas
  $\gamma$ such that the triple $(\beta,\gamma,p)$ satisfies
  Condition~\eqref{eq:bedingung} and
  $(\exists^{(q,p)} x\colon\beta)\Longleftrightarrow\gamma$.
\end{lem}

The proof of this lemma requires several claims and definitions that
we demonstrate first, the actual proof of Lemma~\ref{LB'} can be found
on page~\pageref{Page-proof-LB'}. Its idea is to split the integers
into finitely many intervals (depending on the set of terms that
appear in $\beta$) and to express the number (modulo $p$) of witnesses
for $\beta$ in any such interval by a quantifier-free formula. The
claims below consider different types of such intervals.

Let $N=\lcm\bigl(\MOD(\beta)\bigr)$.  Let $T$ be the set of all pairs
$(a,t)$ such that $\beta$ contains an atomic formula of the form
$ax<t$ or $t<ax$ with $a>0$ (if no such formula exists, set
$T=\bigl\{(1,0)\bigr\}$).

Let $S$ be some non-empty subset of $T$ and let $\prec$ be a strict
linear order on~$S$. We call an assignment $f$ \emph{consistent with
  $(S,\prec)$} if the following hold:
\begin{itemize}
\item $\frac{f(s_1)}{a_1} < \frac{f(s_2)}{a_2} \iff
  (a_1,s_1)\prec(a_2,s_2)$ for all $(a_1,s_1),(a_2,s_2)\in S$
\item for all $(a_1,t_1)\in T$, there exists $(a_2,s_2)\in S$ with
  $\frac{f(t_1)}{a_1}=\frac{f(s_2)}{a_2}$.
\end{itemize}
In the following, let
$S=\bigl\{(a_1,s_1),(a_2,s_2),\dots,(a_n,s_n)\bigr\}$ with
$(a_1,s_1)\prec (a_2,s_2)\prec \dots \prec (a_n,s_n)$. Then any
assignment $f$ that is consistent with $(S,\prec)$ divides $\bZ$ into
the open intervals\footnote{Of course, these intervals are considered
  as sets of integers so that the terms ``open'' and ``closed'' are
  to be understood as ``excluding / including the given bounds if they
  happen to be integers''.}
$\left(-\infty,\frac{f(s_1)}{a_1}\right)$,
$\left(\frac{f(s_i)}{a_i},\frac{f(s_{i+1})}{a_{i+1}}\right)$ for
$1\le i<n$, and $\left(\frac{f(s_n)}{a_n},\infty\right)$, and the
(singleton) closed intervals
$\left[\frac{f(s_j)}{a_j},\frac{f(s_j)}{a_j}\right]$ for
$1\le j\le n$. The following formulas describe (modulo $p$) the number
of witnesses for $\beta$ in these intervals (for $0\le r<p$):
  \begin{align*}
    \beta_{0,r} &= \exists^{(r,p)}x\colon (a_1x<s_1 \land \beta) &
    \beta_{n,r} &= \exists^{(r,p)}x\colon (s_n<a_nx\land\beta)\\
    \beta_{i,r} &= \exists^{(r,p)}x\colon (s_i<a_ix \land a_{i+1} x < s_{i+1} 
                                         \land \beta)&
    \beta'_{j,r} &= \exists^{(r,p)}x\colon (a_jx=s_j\land\beta)
  \end{align*}
Now consider the formula
  \[
    \varphi^\prec =
    \bigvee
     \left(
       \bigwedge_{0\le i\le n}   \beta_{i,r_i}\\
       \land \displaystyle \bigwedge_{1\le i\le n} \beta'_{i,r_i'}
     \right)
  \]
where the disjunction extends over all tuples
$(r_0,r_1,\dots,r_n,r_1',r_2'\dots,r_n')$ of integers from the set
$\{0,1,\dots,p-1\}$ that, modulo $p$, sum up to $q$.  For any
assignment~$f$ consistent with $(S,\prec)$, we get
$f\bigl(\exists^{(q,p)}x\colon\beta\bigr) =
f\bigl(\varphi^\prec\bigr)$.  In order to construct $\gamma$ as
claimed in Lemma~\ref{LB'}, we next eliminate the counting quantifiers
from the formulas $\beta_{0,r}$, $\beta_{i,r}$, $\beta_{n,r}$, and
$\beta'_{j,r}$. In this elimination procedure (detailed in the
following claims), we will assume the assignment $f$ to be consistent
with $(S,\prec)$.

\begin{clm}\label{Cl-0,n+1}
  Let $0\le r<p$. There exists a Boolean combination
  $\gamma^\prec_{0,r}$ of $x$-free atomic formulas such that the
  triple $(\beta,\gamma^\prec_{0,r},p)$ satisfies
  Condition~\eqref{eq:bedingung} and
  $f(\beta_{0,r}) = f(\gamma^\prec_{0,r})$ for all assignments~$f$
  that are consistent with~$(S,\prec)$.
\end{clm}

\begin{proof}
  Let $f$ be an assignment that is consistent with $(S,\prec)$.  Let
  $b\in\bZ$ with $a_1b<f(s_1)$. For all $(a,t)\in T$, we have
  $b<\frac{f(s_1)}{a_1}\le \frac{f(t)}{a}$ and therefore
  $a(b-N)<ab<f(t)$. Consequently, $b$ and $b-N$ satisfy the same
  inequalities from $\beta$. Since $N$ is a multiple of all moduli
  appearing in $\beta$, the same holds for all modulo
  constraints. Hence we obtain
    \[
      f_{x/b}(\beta)=f_{x/b-N}(\beta)\,.
    \]
  Consequently, there are infinitely many $b\in\bZ$ satisfying
  $f_{x/b}(a_1x<s_1\land \beta)=\true$ or none. For $r\neq0$, we can
  therefore set $\gamma^\prec_{0,r}=(0<0)$ ensuring
  Condition~\eqref{eq:bedingung} for the triple
  $(\beta,\gamma_{0,r}^\prec,p)$. It remains to consider the case
  $r=0$. Note that
    \[
      f(\beta_{0,0})=f\bigl(\exists^{(0,p)}x\colon (a_1x<s_1\land\beta)\bigr)=
        f\bigl(\neg\exists x\colon(a_1x<s_1\land\beta)\bigr)
    \]
   since, if any, infinitely many integers $b<\frac{f(s_1)}{a_1}$
  satisfy $f_{x/b}(\beta)=\true$.  Let $\alpha$ be the formula
  obtained by Lemma~\ref{LB} from the formula
  $\exists x\colon(a_1x<s_1\land\beta)$ and set
  $\gamma_{0,0}^\prec=\lnot\alpha$.  Since $a_1x<s_1$ or $s_1<a_1x$ is
  an atomic formula from $\beta$, we get
  $\COEFF(\beta)=\COEFF(a_1x<s_1\land\beta)$ and similarly for
  $\CONST$ and $\MOD$. Hence the triple $(\beta,\alpha,p)$ and therefore
  $(\beta,\gamma^\prec_{0,0},p)$ satisfies
  Condition~\eqref{eq:bedingung}. 
\end{proof}

\noindent 
Symmetrically, we also get the following:

\begin{clm}\label{Cl-n+1}
  Let $0\le r<p$. There exists a Boolean combination
  $\gamma^\prec_{n,r}$ of $x$-free atomic formulas such that the
  triple $(\beta,\gamma^\prec_{n,r},p)$ satisfies
  Condition~\eqref{eq:bedingung} and
  $f(\beta_{n,r}) = f(\gamma^\prec_{n,r})$ for all assignments $f$
  that are consistent with~$(S,\prec)$.
\end{clm}

We next want to eliminate the initial quantifier $\exists^{(r,p)}$
from $\beta_{i,r}$ for $1\le i< n$, i.e., we consider the integers in
the open interval
$\left(\frac{f(s_i)}{a_i},\frac{f(s_{i+1})}{a_{i+1}}\right)$. To get
the idea of the rather long proof, consider the formula
\[
  \exists^{(0,2)}x\colon y<x<z\land x\equiv_3 y+z
\]
and assume that the assignment $f$ satisfies $f(y)<f(z)$. Then the
witnesses for $\varphi:= (x\equiv_3 y+z)$ in the interval
$\bigl(f(y),f(z)\bigr)$ are $3$-periodic. Consequently, any
subinterval of length $6=3\cdot 2$ contains an even number of
witnesses for $\varphi$. It follows that we only need to count the
number of witnesses of $\varphi$ in the interval
$\bigl(f(y),f(y)+b\bigr)$ where $1\le b\le 6$ is the unique number
satisfying $6\mid f(z)-b$ (since then the length of the interval
$\bigl[f(y)+b,f(z)\bigr)$ is a multiple of $6$).

The main additional difficulty in the following proof is based on the
occurrence of subformulas of the form $ax<t$ for $a>0$.

\begin{clm}\label{Cl-i}
  Let $1\le i<n$ and $0\le r<p$. There exists a Boolean
  combination~$\gamma^\prec_{i,r}$ of $x$-free atomic formulas such
  that the triple $(\beta,\gamma^\prec_{i,r},p)$ satisfies
  Condition~\eqref{eq:bedingung} and
  $f(\beta_{i,r}) = f(\gamma^\prec_{i,r})$ for all assignments~$f$
  consistent with~$(S,\prec)$.
\end{clm}

\begin{proof}
  Let $f$ be any assignment that is consistent with~$(S,\prec)$ and
  let $W\subseteq\bZ$ be the set of witnesses for $\beta$, i.e.,
    \[
      W=\{w\in\bZ\mid f_{x/w}(\beta)=\true\}\,.
    \]
  Furthermore, we write $I$ for the interval
  $\left(\frac{f(s_i)}{a_i},\frac{f(s_{i+1})}{a_{i+1}}\right)$. Our
  task is to express, by a quantifier-free formula and irrespective of
  the concrete $(S,\prec)$-consistent assignment $f$, that
  $|I\cap W|\equiv_p r$ holds.

  We first split the interval $I$ into an initial segment of length
  $\le pN$ and subsequent subintervals of length $pN$ each. To this
  aim, let $b$ be the unique integer from the set
  $\{1,2,\dots,a_ia_{i+1}pN\}$ with
    \[
      b\equiv_{a_ia_{i+1}pN} a_i f(s_{i+1}) - a_{i+1} f(s_i)\,.
    \]
  Since $N$ is the least common multiple of $\MOD(\beta)$, this is
  equivalent to requiring that the formula
    \[
      \bigwedge_{m\in\MOD(\beta)} b\equiv_{a_ia_{i+1}pm} a_i s_{i+1} - a_{i+1} s_i
    \]
  evaluates to $\true$ under the assignment $f$.
  Note that $a_ia_{i+1}pN$ divides $a_i f(s_{i+1}) - a_{i+1} f(s_i) - b$, hence
    \[
      K:= \frac{a_i f(s_{i+1}) - a_{i+1} f(s_i) - b}{a_ia_{i+1}pN}
    \]
  is an integer. Even more,
  $\frac{f(s_i)}{a_i}<\frac{f(s_{i+1})}{a_{i+1}}$ implies
  $b\le a_i f(s_{i+1}) - a_{i+1} f(s_i)$ and therefore $K\in\bN$.  Now
  we define the following intervals:
  \begin{itemize}
  \item $I_0=\left(\frac{f(s_i)}{a_i},\frac{f(s_i)}{a_i}+\frac{b}{a_i a_{i+1}}
    \right)$
  \item
    $J_k=\left[\frac{f(s_i)}{a_i}+\frac{b}{a_i a_{i+1}}+k\cdot
      pN,\frac{f(s_i)}{a_i}+\frac{b}{a_i a_{i+1}}+(k+1)\cdot
      pN\right)$ for $0\le k<K$
  \end{itemize}
  Note that these intervals form a partition of the interval $I$.

  Let $c\in \bZ$ with
  $\frac{f(s_i)}{a_i}<c<c+N<\frac{f(s_{i+1})}{a_{i+1}}$, i.e.,
  $c,c+N\in I$. Since $f$ is consistent with $(S,\prec)$, for any
  $(a,t)\in T$, we have $\frac{f(t)}{a} \le \frac{f(s_i)}{a_i}<c<c+N$
  or $c<c+N<\frac{f(s_{i+1})}{a_{i+1}}\le\frac{f(t)}{a}$. Hence $c$
  and $c+N$ satisfy the same inequalities from $\beta$. Since $N$ is a
  multiple of all moduli appearing in $\beta$, it follows that $c$ and
  $c+N$ also satisfy the same modulo constraints from $\beta$. Hence
  we get
    \[
      f_{x/c}(\beta)=f_{x/c+N}(\beta)\,.
    \]
  It follows that the set $W$ of witnesses for $\beta$ within the
  interval $I$ is $N$-periodic. Since the interval $J_k\subseteq I$ is
  of length $pN$, it follows that $|J_k\cap W|\equiv_p 0$ for all
  $0\le k<K$. Consequently,
    \begin{align*}
      |I\cap W| &= |I_0\cap W| + \sum_{0\le k<K}|J_k\cap W|\\
      &\equiv_p |I_0\cap W|\,.
    \end{align*}
  It remains to construct a formula expressing that the interval $I_0$
  has, modulo $p$, $r$ witnesses for $\beta$.

  To characterise the elements of $I_0$, let $e\in\bZ$ be
  arbitrary. By the definition of $I_0$, we have $e\in I_0$ iff
  $a_{i+1} f(s_i) < a_i a_{i+1} e < a_{i+1} f(s_i)+b$. This is clearly
  equivalent to $0 < a_i e - f(s_i) <
  \frac{b}{a_{i+1}}$. Equivalently, there exists an integer $d$ with
    \[
      0 < d \le \left\lfloor\frac{b-1}{a_{i+1}}\right\rfloor \text{ and }
      e=\frac{f(s_i)+d}{a_i}\,.
    \]
  Set
  $M=\Bigl\{1,2,\dots,\bigl\lfloor\frac{b-1}{a_{i+1}}\bigr\rfloor\Bigr\}$.
  Then we showed
    \[
      I_0 =\left\{\left.\frac{f(s_i)+d}{a_i}\ \right|
        d\in M, f(s_i)+d\equiv_{a_i}0\right\}\,.
    \]
  Now let $d\in M$ with $f(s_i+d)\equiv_{a_i}0$ be arbitrary and set
  $e=\frac{f(s_i+d)}{a_i}$. Then we have
    \[
      f_{x/e}(a_i x) = a_i e = f_{x/e}(s_i+d)\,.
    \]
  Hence, by Lemma~\ref{L-neu}, we get
    \[
      f_{x/e}(\beta) = f_{x/e}(\beta_{a_i,s_i+d}) = f(\beta_{a_i,s_i+d})
    \]
  where the last equality holds since $\beta_{a_i,s_i+d}$ is
  $x$-free. It follows that $e\in W$ iff
  $f(\beta_{a_i,s_i+d})=\true$. Hence we showed that $I_0\cap W$ is the
  set of fractions $\frac{f(s_i+d)}{a_i}$ for $d\in M$ with
  $f(s_i+d)\equiv_{a_i}0$ and $f(\beta_{a_i,s_i+d})=\true$. We
  consequently get
    \begin{align*}
      |I_0\cap W| &=
        \Biggl| \biggl\{ \frac{f(s_i+d)}{a_i} \colon d\in M,
                                   f(s_i+d)\equiv_{a_i}0, 
                                   f(\beta_{a_i,s_i+d})=\true
                \biggr\}
        \Biggr|\\
      &=
        \biggl| \Bigl\{ d\in M \colon 
                                   f\bigl(s_i+d\equiv_{a_i}0 \land
                                        \beta_{a_i,s_i+d}\bigr)=\true
                \Bigr\}
        \biggr|\,.
    \end{align*}
  It follows that $|I_0\cap W|\equiv_p r$ iff the following formula
  $\gamma_{i,r}^\prec$ holds under the assignment $f$:
    \[
      \bigvee_{1\le b\le a_ia_{i+1}pN}
        \left(
          \begin{array}{ll}
            & \displaystyle
              \bigwedge_{m\in\MOD(\beta)} b\equiv_{a_ia_{i+1}pm} a_i s_{i+1} - a_{i+1} s_i\\
            \land
              & \displaystyle
            \bigvee_{\genfrac{}{}{0pt}{}{W_0\subseteq M}{|W_0|\equiv_p r}}
            \left(
                \begin{array}{ll}
                  &\displaystyle
                  \bigwedge_{d\in W_0}
                    \bigl( s_i+d\equiv_{a_i}0 \land
                           \beta_{a_i,s_i+d}
                    \bigr)\\
                  \land&\displaystyle
                  \bigwedge_{d\in M\setminus W_0}
                    \lnot
                    \bigl( s_i+d\equiv_{a_i}0 \land
                           \beta_{a_i,s_i+d}
                    \bigr)
                \end{array}
                \right)
          \end{array}
        \right)
    \]
  
  We finally verify Condition~\eqref{eq:bedingung} for the triple
  $(\beta,\gamma_{i,r}^\prec,p)$. Note that any element of
  $\COEFF(\gamma_{i,r}^\prec)$ or $\CONST(\gamma_{i,r}^\prec)$ appears
  in a subformula of the form $\beta_{a_i,s_i+d}$ for some integer
  $d\in M$ and therefore
    $1\le d\le \frac{b-1}{a_{i+1}}<a_ipN$. Hence
  $\COEFF(\gamma_{i,r}^\prec)\subseteq\COEFF_p(\beta)$ and
  $\CONST(\gamma_{i,r}^\prec)\subseteq\CONST_p(\beta)$ follow from
  Lemma~\ref{L-neu}.

  Now let $p_1\in\MOD(\gamma_{i,r}^\prec)$. There are three cases to be
  considered:
  \begin{itemize}
  \item $p_1=a_ia_{i+1}pm$ for some $m\in\MOD(\beta)$. Then $p_1\in\MOD_p(\beta)$.
  \item $p_1=a_i$. Then $p_1\in\COEFF(\beta)\subseteq\MOD_p(\beta)$
  \item $p_1\in\MOD(\beta_{a_i,s_i+d})$ for some integer
    $d$ with
      \[
        1\le d \le a_ipN\,.
      \]
    Then, by Lemma~\ref{L-neu}, $p_1\in\MOD_p(\beta)$.  
  \end{itemize}
  Thus, indeed, $\MOD(\gamma_{i,r}^\prec)\subseteq\MOD_p(\beta)$ which
  finishes the proof of Claim~\ref{Cl-i}.
\end{proof}

\begin{clm}\label{Cl-i'}
  Let $1\le j\le n$ and $0\le r<p$. There exists a Boolean combination
  $\delta^\prec_{j,r}$ of $x$-free atomic formulas such that
  $(\beta,\delta^\prec_{j,r},p)$ satisfies
  Condition~\eqref{eq:bedingung} and, for all assignments $f$ (even
  those that are not consistent with~$(S,\prec)$),
  $f(\beta'_{j,r}) = f(\delta^\prec_{j,r})$.
\end{clm}

\begin{proof}
  Since the term $s_j$ is $x$-free, there can be at most one witness
  for the formula $a_jx=s_j\land\beta$ (which is the quantifier-free
  part of the formula $\beta'_{j,r}$). For $r>1$, we therefore set
  $\delta_{j,r}^\prec=(0<0)$.

  For the same reason, we obtain
    \[
      \exists^{(1,p)}x\colon(a_jx=s_j\land\beta)
      \iff \exists x\colon(a_jx=s_j\land\beta)\,.
    \]
  Hence, we obtain the formula $\delta_{j,1}^\prec$ from
  Lemma~\ref{LB}. Since precisely one of the formulas
  $\delta_{j,r}^\prec$ must hold, we can set
  $\delta_{j,0}^\prec=\bigwedge_{0<r<p}\lnot\delta_{j,r}^\prec$ (which
  is equivalent to $\lnot\delta_{j,1}^\prec$).
\end{proof}

Having shown all these claims, we can now use them to finally prove
Lemma~\ref{LB'}.

\begin{proof}[Proof of Lemma~\ref{LB'}]\label{Page-proof-LB'}
  Let $S\subseteq T$ be some non-empty subset of $T$ and let $\prec$
  be a strict linear order on $S$. As above, we let
  $S=\bigl\{(a_1,s_1),\dots,(a_n,s_n)\bigl\}$ with $(a_1,s_1)\prec
  (a_2,s_2)\prec \dots\prec (a_n,s_n)$. Then set
    \begin{align*}
      \gamma^\prec=
      \bigvee
      \left(
       \bigwedge_{0\le i\le n+1}  \gamma_{i,r_i}^\prec
       \land
       \bigwedge_{1\le j\le n}  \delta_{j,r_i'}^\prec
      \right)
    \end{align*}
  where the disjunction extends over all tuples
  $(r_0,r_1,\dots,r_{n+1},r_1',r_2'\dots,r_n')$ of natural numbers
  from $\{0,1,\dots,p-1\}$ with
  $ \sum_{0\le i\le n+1}r_i+\sum_{1\le i\le n}r_i'\equiv_p q$. The
  above claims imply $f(\varphi^\prec)=f(\gamma^\prec)$ for all
  assignments $f$ that are consistent with~$(S,\prec)$. Furthermore,
  $\gamma^\prec$ is a Boolean combination of atomic formulas and the
  triple $(\beta,\gamma^\prec,p)$ satisfies
  Condition~\eqref{eq:bedingung}.

  Next consider the formula
    \[
      \alpha^\prec=
      \bigwedge_{1\le i<n}a_{i+1}s_i<a_is_{i+1}\land
      \bigwedge_{(a,t)\in T}\bigvee_{1\le i\le n}a_it=as_i\,.
    \]
  Then, for any assignment $f$, we have $f(\alpha^\prec)=\true$ if and
  only if $f$ is consistent with~$(S,\prec)$. Since $\alpha^\prec$ is a
  Boolean combination of formulas of the form $a's<at$\footnote{Write
    $\lnot a_it<as_i\land\lnot a_it>as_i$ for $a_it=as_i$.} with
  $(a,s),(a',t)\in T$, the triple $(\beta,\alpha^\prec,p)$ satisfies
  Condition~\eqref{eq:bedingung}.

  Finally, let
  \[
     \gamma=\bigvee_{(*)}(\alpha^\prec\land\gamma^\prec)
   \]
   where the disjunction $(*)$ extends over all strict linear orders
  $\prec$ on some non-empty subset of~$T$.
\end{proof}

Lemmas~\ref{LB} and \ref{LB'} above show how to eliminate a quantifier
in front of a quantifier-free formula and analyses the sets of
coefficients, constants, and moduli appearing in this process. The
following proposition summarises these results and provides bounds on
the maximal coefficients etc. Recall that
$\Prod(\varphi)=\COEFF(\varphi)\cup\MOD(\varphi)$.

\begin{prop}\label{P'}
  Let $x$ be a variable and $\alpha$ a Boolean combination of atomic
  formulas. Let furthermore $E=\exists$ or $E=\exists^{(q,p)}$ for
  some $0\le q<p$ and $2\le p$. Then there exists a Boolean combination
  $\gamma$ of $x$-free atomic formulas such that $(E x\colon\alpha)
  \Longleftrightarrow \gamma$.
  Furthermore, we have the following:
    \begin{align*}
      \max\Prod(\gamma) & \le \max\Prod(Ex\colon\alpha)^4\\
      \max\CONST(\gamma) &
          \le \max\CONST(Ex\colon\alpha)\cdot 16^{\max\Prod(Ex\colon\alpha)}
    \end{align*}
\end{prop}

\begin{proof}
  If $E=\exists$, set $p=1$.
  Without changing the sets of coefficients etc., we can transform
  $\alpha$ into an equivalent Boolean combination~$\beta$ of
  $x$-separated atomic formulas. By Lemma~\ref{LB} or \ref{LB'}, there
  exists a Boolean combination $\gamma$ of  $x$-free atomic formulas with
  $(Ex\colon \alpha)\Longleftrightarrow\gamma$ such that the triple
  $(\alpha,\gamma,p)$ satisfies Condition~\eqref{eq:bedingung}.

  Note that $\max\COEFF(\alpha),\max\MOD(\alpha)\le\max\Prod(\alpha)$.
  From $\COEFF(\gamma)\subseteq\COEFF_p(\alpha)$ and
  $\MOD(\gamma)\subseteq\MOD_p(\alpha)$, we can therefore infer
    \begin{align*}
      \max\COEFF(\gamma) &\le \max\COEFF_p(\alpha)
                          \le 2\cdot\max\COEFF(\alpha)^2 \\
                         & \le \max\COEFF(\alpha)^3\\
                         &\le \max\Prod(Ex\colon\alpha)^4
      \intertext{ and}
      \max\MOD(\gamma)   &\le \max\MOD_p(\alpha)\\
                         &\le \max\COEFF(\alpha)^2
                                 \cdot\max\MOD(Ex\colon\alpha)^2\\
                         &\le \max\Prod(Ex\colon\alpha)^4\,.
    \end{align*}
  Consequently, $\max\Prod(\gamma)\le\max\Prod(Ex\colon\alpha)^4$.
  
  From $\CONST(\gamma)\subseteq\CONST_p(\alpha)$, we can infer
    \begin{align*}
      \max\CONST(\gamma) &\le 
         2\cdot\max\COEFF(\alpha) \cdot \max\CONST(\alpha)
         +\max\COEFF(\alpha)^2\cdot p\cdot\lcm\bigl(\MOD(\alpha)\bigr)\\
        &\le 
          \max\Prod(\alpha)^2 \cdot (\max\CONST(\alpha)
          +p\cdot \lcm\{1,2,\dots,\max\MOD(\alpha)\})\,.\\
      \intertext{Since $\lcm\{1,2,\dots,n\}\le 4^{n-1}$ by \cite{Nai82}, we can continue}
        &\le 
          \max\Prod(\alpha)^2 \cdot (\max\CONST(\alpha)
          +p\cdot 4^{\max\Prod(\alpha)})\\
        &\le 
          2^{\max\Prod(\alpha)} \cdot (\max\CONST(\alpha)
          +2^{p+2\cdot\max\Prod(\alpha)})
      \ (\text{provided }\max\Prod(\alpha)\neq 3)\\
        &\le 
          \max\CONST(\alpha)
          \cdot 2^{p+3\cdot\max\Prod(\alpha)}\\
        &\le 
          \max\CONST(\alpha)
          \cdot 2^{4\cdot\max\Prod(Ex\colon\alpha)}\\
        &=
          \max\CONST(Ex\colon\alpha)\cdot 16^{\max\Prod(Ex\colon\alpha)}\,.
    \end{align*}
    If $\max\Prod(\alpha)=3$, we get
    $\max\Prod(\alpha)^2 \cdot (\max\CONST(\alpha) +p\cdot
    4^{\max\Prod(\alpha)})\le \max\CONST(\alpha) \cdot
    2^{p+3\cdot\max\Prod(\alpha)}$ as well since $2p\le 2^p$ so that
    the desired estimation holds in this case, too.
\end{proof}

Now, by induction on the quantifier depth we can obtain the following theorem.

\begin{thm}\label{thm:modelim}
  Let $\varphi\in\FOMOD$ be a formula of
  quantifier-depth~$d$. There exists an equivalent Boolean
  combination~$\gamma$ of atomic formulas with
    \begin{align*}
      \max\Prod(\gamma) & \le \max\Prod(\varphi)^{4^d}\,,\text{ and }\\
      \max\CONST(\gamma) & 
           \le 2^{(\max\Prod(\varphi))^{4^d}}\cdot\max\CONST(\varphi)\,.
    \end{align*}
\end{thm}

\begin{proof}
  The proof proceeds by induction on $d$. For $d=0$, the claim is
  trivial since then, we can set $\gamma=\varphi$. Now suppose the
  theorem has been shown for formulas of quantifier-depth~$<d$.

  So let $\varphi = Ex\colon\psi$ where $E=\exists$ or
  $E=\exists^{(q,p)}$ for some $0\le q<p$ and the formula $\psi$ has
  quantifier-rank~$<d$. If $E=\exists$, set $p=1$. Then, by the
  induction hypothesis, there exists a Boolean combination~$\alpha$ of
  atomic formulas such that $\psi\Longleftrightarrow\alpha$,
    \begin{align*}
      \max\Prod(\alpha) & \le (\max\Prod(\psi))^{4^{d-1}}\text{ and}\\
      \max\CONST(\alpha)& \le 2^{(\max\Prod(\psi))^{4^{d-1}}}\cdot\max\CONST(\psi)\,.
    \end{align*}
  By Prop.~\ref{P'}, we find a Boolean combination $\gamma$ of atomic
  formulas such that the following hold:
  \begin{itemize}
  \item
    $\gamma \Longleftrightarrow Ex\colon\alpha \Longleftrightarrow
    Ex\colon\psi=\varphi$
  \item $\max\Prod(\gamma)\le\max\Prod(Ex\colon\alpha)^4$
  \item
    $\max\CONST(\gamma) \le \max\CONST(Ex\colon\alpha)\cdot
    16^{\max\Prod(Ex\colon\alpha)}$
  \end{itemize}

  Note that $\max\Prod(Ex\colon\alpha)$ is the maximum of
  $p\le p^{4^{d-1}}$ and
  $\max\Prod(\alpha)\le\max\Prod(\psi)^{4^{d-1}}$.  Similarly, the
  maximum of $p$ and $\max\Prod(\psi)$ is equal to
  $\max\Prod(Ex\colon\psi)$. Therefore we get
  $\max\Prod(Ex\colon\alpha)\le\max\Prod(Ex\colon\psi)^{4^{d-1}}$.
  Hence
    \begin{align*}
      \max\Prod(\gamma) 
      &\le \max\Prod(Ex\colon\alpha)^4\\
      &\le (\max\Prod(Ex\colon\psi)^{4^{d-1}})^4 \\
      &= \max\Prod(\varphi)^{4^d}\,.
    \end{align*}
  Before we prove the desired upper bound for $\max\CONST(\gamma)$,
  note the following for all $n\ge2$ and $d\ge1$:
    \begin{align*}
      \log_2\bigl(16^{n^{4^{d-1}}}\cdot 2^{n^{4^{d-1}}}\bigr)
      &=4\cdot n^{4^{d-1}} + n^{4^{d-1}}\\
      &\le n^3\cdot n^{4^{d-1}}\\
      &= n^{3+ 4^{d-1}}\\
      &\le n^{4^d}\,.
    \end{align*}
  With $n=\max\Prod(\varphi)$, we therefore obtain
    \begin{align*}
      \max\CONST(\gamma)
       &\le 16^{\max\Prod(Ex\colon\alpha)}\cdot\max\CONST(Ex\colon\alpha)\\
       &\le 16^{\max\Prod(\varphi)^{4^{d-1}}}\cdot 
          2^{\max\Prod(\varphi)^{4^{d-1}}}\cdot\max\CONST(\varphi)\\
       &\le2^{\max\Prod(\varphi)^{4^d}}\cdot\max\CONST(\varphi)\,.\qedhere
    \end{align*}
\end{proof}

Using Proposition~\ref{P-FOMOD-to-FOMOD-}, the extension of the above
result to the larger logic $\FOMODmany$ follows immediately.
\begin{cor}
  Let $\varphi\in\FOMODmany$ be a formula of quantifier-depth $d$. There
  exists an equivalent Boolean combination $\gamma$ of atomic formulas
  with
    \begin{align*}
      \max\Prod(\gamma) & \le \max\Prod(\varphi)^{4^d}\text{ and }\\
      \max\CONST(\gamma) & 
           \le 2^{(\max\Prod(\varphi))^{4^d}}\cdot\max\CONST(\varphi)\,.
    \end{align*}
\end{cor}
If we allow the threshold counting quantifiers $\exists^{\ge c}$ and
$\exists^{=c}$, the result gets a bit weaker since we have to replace
the exponent $d$ in the above bounds by a polynomial in the size of
$\varphi$. To see this, let $\varphi\in\fullLogic$. Then, by
Proposition~\ref{P-CMOD-to-FOMOD}, it can be transformed in polynomial
time into an equivalent formula $\varphi'$ from $\FOMODmany$. The
quantifier depth $d'$ of $\varphi'$ is bounded by the size of
$\varphi'$ and therefore polynomial in the size of $\varphi$. Now we
can resort to the above corollary and obtain
\begin{cor}
  Let $\varphi\in\fullLogic$ be a formula. There exists an
  equivalent Boolean combination~$\gamma$ of atomic formulas with
    \begin{align*}
      \max\Prod(\gamma) & \le \max\Prod(\varphi)^{4^{\mathrm{poly}(|\varphi|)}}\text{ and }\\
      \max\CONST(\gamma) & 
           \le 2^{(\max\Prod(\varphi))^{4^{\mathrm{poly}(|\varphi|)}}}\cdot\max\CONST(\varphi)\,.
    \end{align*}
\end{cor}

\section{An efficient decision procedure}\label{SS-decision}

Let $\varphi(x)$ be a Boolean combination of formulas with a single
free variable. To determine validity of the formula
$\exists x\colon\varphi$, one has to check, for all integers $n\in\bZ$,
whether $\varphi(n)$ holds. The following lemma reduces this infinite
search space to a finite one that is exponential in the coefficients
and moduli as well as linear in the constants from $\varphi$.

\begin{lem}\label{L-AB}
  Let $A\ge6$ and $B\ge0$. Let $x$ be a variable and $\gamma$ a
  Boolean combination of atomic formulas of the form $ax>b$, $ax<b$,
  and $cx\equiv_h d$ with $a,b,c,d\in\bZ$, $h\ge2$, $|a|,h<A$, and
  $|b|<B$. Then $\exists x\colon\gamma$ is equivalent to
  $\exists x\colon \bigl(|x|\le A^{A^5}\cdot B \land \gamma\bigr)$.
\end{lem}

\begin{proof}
  Since $h<A$, we can assume that $0\le c,d<A$ for all formulas of the
  form $cx\equiv_h d$. We can also assume that $\gamma$ is in
  negation normal form, i.e., only atomic formulas are negated. We
  make the following replacements:
  \begin{center}
    \begin{tabular}{rcl}
      $\neg(ax > b)$ & is replaced by & $ax < b+1$\\
      $\neg(cx \equiv_h d)$ & is replaced by & 
             $\bigvee_{0\le d'<h, d\neq d'} cx\equiv_h d'$\\
      $ax>b$ & is replaced by & $-ax< -b$
    \end{tabular}
  \end{center}
  As a result, $\gamma$ is equivalent to a formula in disjunctive
  normal form, without negations, and with atomic formulas of the form
  $ax<b$ and $cx\equiv_h d$ with $0\le c,d,|a|,h< A$ and $|b|\le
  B$. Hence $\gamma\Longleftrightarrow\bigvee_{1\le i\le n}\delta_i$
  where each of the formulas $\delta_i$ is a conjunction of atomic
  formulas of the allowed form. Consequently, $\exists x\colon\gamma$
  is equivalent to $\bigvee_{1\le i\le n}\exists x\colon\delta_i$.

  Consider one such conjunction~$\delta_i$. Note that it contains at
  most $A^3$ many atomic formulas of the form $cx\equiv_h d$ since
  $0\le c,d,h<A$. For any such atomic formula, introduce a new
  variable $y$ and replace $cx\equiv_h d$ by $cx-hy=d$. Then $\delta_i$
  is equivalent to $\exists \bar y\colon\delta'_i$ where $\delta'_i$ is a
  conjunction of formulas of the form $cx-hy=d$ and $ax<b$ with $0\le
  c,h,d,|a|<A$ and $|b|\le B$ and $\bar y$ is a sequence of at most
  $A^3$ variables.

  Let $M$ be the maximal absolute value of the determinant of an
  $(m\times m)$-matrix with $m\le A^3+2$, where the first $m-1$
  columns contain entries of absolute value at most $A$ and the
  entries in the last column have absolute value at most $B$. Then it
  is not hard to determine that
  \[
    M \le (A^3+2)!\cdot A^{A^3+1}\cdot B\,.
  \]
  Now the main theorem of \cite{VonZurGathenSieveking} implies that
  the formula $\exists x,\bar y\colon\delta'$ is equivalent to the
  existence of a solution $(x,\bar y)$ of $\delta'$ where the absolute
  value of every entry is at most
  \begin{align*}
    (A^3+2)\cdot M &\le A^4\cdot(A^4)!\cdot A^{A^4}\cdot B\\
      &\le A^4\cdot (A^4)^{A^4}\cdot A^{A^4}\cdot B\\
      &\le A^{4+5\cdot A^4}\cdot B \le A^{A^5}\cdot B\,.
  \end{align*}
  In summary, we get
  \begin{align*}
    \exists x\colon\gamma
       &\Longleftrightarrow \bigvee \exists x\colon\delta_i\\
       &\Longleftrightarrow \bigvee \exists x\exists\bar y\colon\delta_i'\\
       &\Longleftrightarrow \bigvee \exists x\colon\bigl(|x|\le A^{A^5}\cdot B\land 
                                     \exists\bar y\colon\delta_i'\bigr)\\
       &\Longleftrightarrow \exists x\colon\bigl(|x|\le A^{A^5}\cdot B\land 
                                     \bigvee\delta_i\bigr)\\
       &\Longleftrightarrow \exists x\colon\bigl(|x|\le A^{A^5}\cdot B\land \gamma\bigr)
  \end{align*}
  where all disjunctions extend over $1\le i\le n$.
\end{proof}

The core of the above lemma is the reduction of the search space for
closed formulas of the form $\exists x\colon\varphi(x)$ with $\varphi$
quantifier-free. The following corollary provides an analogous
reduction for arbitrary formulas $\varphi(x)$. In addition, we allow
the formula $\varphi$ to have further free variables
$y_1,\dots,y_\ell$ that are handled as parameters.

\begin{cor}\label{C1}
  There exists $\kappa\ge2$ with the following property.  Let $d\ge1$
  and consider a formula $\varphi(x,y_1,\dots,y_\ell)$ from $\FOMOD$
  of quantifier-depth at most $d$. Let $n_1,\dots,n_\ell\in\bZ$ with
  $|n_i|\le N$. Then the closed formula
  $\exists x\colon\varphi(x,n_1,\dots,n_\ell)$ holds if and only if
  there exists $n\in\bZ$ such that $\varphi(n,n_1,\dots,n_\ell)$ holds
  with

    \[\displaystyle
      |n|\le 2^{\max\Prod(\varphi)^{\kappa^d}}\cdot
                 \max\CONST(\varphi) \cdot N\cdot\max\{1,\ell\}\,.
    \]  
\end{cor}

\begin{proof}
  The implication ``$\Leftarrow$'' is trivial since, if there is a
  small $n$ satisfying $\varphi(x,n_1,\dots,n_\ell)$, then $\exists
  x\colon\varphi(x,n_1,\dots,n_\ell)$ holds.

  Conversely suppose $\exists x\colon\varphi(x,n_1,\dots,n_\ell)$
  holds.  Let $\varphi_{\bar n}=\varphi_{\bar n}(x)$ be the formula
  obtained from $\varphi$ by substituting $n_i$ for $y_i$.  For any
  inequality $s<t$ in $\varphi$, the term $s-t$ contains at most
  $\ell$ of the variables $y_i$, each with a coefficient from
  $\COEFF(\varphi)\subseteq\Prod(\varphi)$. Hence these substitutions
  at most eliminate coefficients, do not change moduli, but can
  increase constants by $N\cdot\ell\cdot\max\Prod(\varphi)$. Hence we
  get
    \begin{align*}
      \max\Prod(\varphi_{\bar n}) &\le\max\Prod(\varphi)\ \text{and}\\
      \max\CONST(\varphi_{\bar n}) 
         &\le \max\CONST(\varphi) + N\cdot\ell\cdot\max\Prod(\varphi)\,.
    \end{align*}
  By Theorem~\ref{thm:modelim}, there exists an equivalent Boolean
  combination $\gamma_{\bar n}$ of atomic formulas with
    \begin{align*}
      \max\Prod(\gamma_{\bar n}) 
         &\le \max\Prod(\varphi)^{4^d}=:A\ \text{and }\\
      \max\CONST(\gamma_{\bar n})
         &\le 2^{\max\Prod(\varphi)^{4^d}}\cdot
               (\max\CONST(\varphi) + N\cdot\ell\cdot\max\Prod(\varphi))\\
         &\le 2^{\max\Prod(\varphi)^{5^d}}\cdot
               \max\CONST(\varphi) \cdot N\cdot\max\{1,\ell\}=:B\,.
    \end{align*}
  From Lemma~\ref{L-AB}, we obtain that there is some $n\in\bZ$ with
  $|n|\le A^{A^5}\cdot B$ such that $\gamma_{\bar n}(n)$ holds. Hence,
  for this $n$, also $\varphi(n,n_1,\dots,n_\ell)$ holds. Note that
  we have (with $p=\max\Prod(\varphi)$)
    \[
      A^{A^5}\le\left(p^{4^d}\right)^{\left(p^{4^d}\right)^5}
       = p^{4^d\cdot p^{5\cdot 4^d}}
       \le p^{p^{5\cdot 4^d+2d}}\le 2^{p^{c^d}} 
       = 2^{\max\Prod(\varphi)^{c^d}}
    \]
  for some $c\ge1$ and therefore
    \begin{align*}
      |n| & \le 2^{\max\Prod(\varphi)^{c^d}}\cdot
                2^{\max\Prod(\varphi)^{5^d}}\cdot \max\CONST(\varphi)
                \cdot N\cdot\max\{1,\ell\}\\
         &\le 2^{\max\Prod(\varphi)^{\kappa^d}}\cdot
               \max\CONST(\varphi) \cdot N\cdot\max\{1,\ell\}
    \end{align*}
  for some $\kappa\ge2$. 
\end{proof}

In the following, we want to prove a similar result for the modulo-counting
quantifier. Recall that $\exists^{(q,p)}x\colon\varphi(x)$ can only
be true if $\varphi$ has only finitely many witnesses, i.e., if the
formula $\exists y\forall x\colon\bigl(\varphi(x)\to |x|\le y\bigr)$
is true.  Applying the above corollary, one finds a finite interval
such that $\varphi$ has infinitely many witnesses iff it has at least
one witness in this interval. In case $\varphi$ has only finitely many
witnesses, then all of them are of bounded absolute value. More
precisely, we get the following.

\begin{lem}\label{L2}
  Let $d\ge1$ and $\kappa\ge2$ be the constant from
  Corollary~\ref{C1}.  Furthermore, let 
  $\varphi=\varphi(x,y_1,\dots,y_\ell)\in\FOMOD$ be a formula of
  quantifier-depth at most~$d$, let $n_1,\dots,n_\ell\in\bZ$ with
  $|n_i|\le N$. Suppose there exist only finitely many $n\in\bZ$ such
  that $\varphi(n,n_1,\dots,n_\ell)$ holds. Then all $n\in\bZ$ such
  that $\varphi(n,n_1,\dots,n_\ell)$ holds satisfy
    \[
    |n|\le 2^{\max\Prod(\varphi)^{\kappa^{d+1}}}
                             \cdot\max\CONST(\varphi)\cdot N\cdot\max\{1,\ell\}\,.
    \]
\end{lem}

\begin{proof}
  Since there are only finitely many $n\in\bZ$ such that
  $\varphi(n,n_1,\dots,n_\ell)$ holds, the closed formula
  \begin{equation}
    \exists y\forall x\colon \bigl(\varphi(x,n_1,\dots,n_\ell)\to |x|\le y\bigr)
    \label{eq:in_L2}
  \end{equation}
  holds. Let $\varphi'$ denote the subformula starting with $\forall
  x$. Note that its quantifier-depth equals $d+1$,
  $\Prod(\varphi)=\Prod(\varphi')$, and
  $\CONST(\varphi)=\CONST(\varphi')$. Hence, by Corollary~\ref{C1},
  the above formula \eqref{eq:in_L2} is equivalent to
    \[
      \exists y\colon\bigl(|y|\le 2^{\max\Prod(\varphi)^{\kappa^{d+1}}}
                             \cdot\max\CONST(\varphi)\cdot N\cdot \max\{1,\ell\}
                      \land \varphi'\bigr)
    \]
  and therefore to
    \[
       \forall x\colon \bigl(\varphi(x,n_1,\dots,n_\ell)\to 
          |x|\le 2^{\max\Prod(\varphi)^{\kappa^{d+1}}}
                             \cdot\max\CONST(\varphi)\cdot N\cdot\max\{1,\ell\}\bigr)\,.
    \]
  Now the claim follows since the formula \eqref{eq:in_L2} and
  therefore this formula holds.  
\end{proof}

\begin{cor}\label{C2}
  Let $d\ge1$ and $\kappa$ be the constant from Corollary~\ref{C1} and
  \[
      C=2^{\max\Prod(\varphi)^{\kappa^{d+1}}}\cdot\max\CONST(\varphi)\cdot
    N\cdot\max\{1,\ell\}\,.
  \]
  Let $\varphi=\varphi(x,y_1,\dots,y_\ell)\in\FOMOD$ be a formula of
  quantifier-depth at most~$d$, let $n_1,\dots,n_\ell\in\bZ$ with
  $|n_i|\le N$. Then
  $\exists^{(q,p)}x\colon\varphi(x,n_1,\dots,n_\ell)$ is true if and
  only if the following hold:
  \begin{enumerate}[(a)]
  \item no integer $n$ with $C<|n|\le C^2$ makes
    $\varphi(n,n_1,\dots,n_\ell)$ true and
  \item $\bigl|\{n\in\bZ\colon |n|\le C\text{ and
    }\varphi(n,n_1,\dots,n_\ell)\text{ is true}\}\bigr|\equiv_p q$\,.
  \end{enumerate}
\end{cor}

\begin{proof}
  We first show that
  $\exists^{(q,p)}x\colon\varphi(x,n_1,\dots,n_\ell)$ is true if and
  only if
  \begin{enumerate}[(a')]
  \item[(a$'$)]
    $\forall x\colon\bigl(\varphi(x,n_1,\dots,n_\ell)\to |x|\le
    C\bigr)$ is true and
  \item[  (b)] $\bigl|\{n\in\bZ\colon |n|\le C\text{ and
    }\varphi(n,n_1,\dots,n_\ell)\text{ is true}\}\bigr|\equiv_p q$\,.
  \end{enumerate}

  Suppose there are infinitely many integers $n$ such that
  $\varphi(n,n_1,\dots,n_\ell)$ holds. Then the formula
  $\exists^{(q,p)}x\colon\varphi$ does not hold. Furthermore,
  statement (a$'$) is false since there are only finitely many
  integers $x$ with $|x|\le C$. Hence, in this case, the equivalence
  holds.

  So it remains to consider the case that there are only finitely many
  integers $n$ such that $\varphi(n,n_1,\dots,n_\ell)$ holds. Then, by
  Lemma~\ref{L2}, all these integers satisfy $|n|\le C$. Consequently,
  statement (a$'$) is true and
    \[
      \bigl\{n\in\bZ\bigm| \varphi(n,n_1,\dots,n_\ell)\text{ holds}\bigr\}
      =
      \bigl\{n\in\bZ\colon |n|\le C\text{ and }
                    \varphi(n,n_1,\dots,n_\ell)\text{ holds}\bigr\}\,.
    \]
  Hence, in this case, $\exists^{(q,p)}x\colon\varphi$ is equivalent
  to statement (b). Since (a$'$) is true in this case,
  we have the equivalence.

  We complete the proof of this corollary by showing that (a)
  and (a$'$) are equivalent. Consider the formula
    \[
       \varphi'=\bigl(\varphi(x,x_1,\dots,x_\ell)\land |x|> C\bigr)\,.
    \]
  Then $\Prod(\varphi')=\Prod(\varphi)$ and
  $\CONST(\varphi')=\CONST(\varphi)\cup\{\pm C\}$ implying
  $\max\CONST(\varphi')=C$. Hence, by Corollary~\ref{C1},
  $\exists x\colon\varphi'$ is equivalent to the existence of
  $n\in\bZ$ satisfying $\varphi(x,n_1,\dots,n_\ell)$ with $C<|n|$ and 
     \begin{align*}
    |n|&\le
        2^{\max\Prod(\varphi')^{\kappa^d}}
        \cdot\max\CONST(\varphi')\cdot N\cdot\max\{1,\ell\}\\
       &=
        2^{\max\Prod(\varphi)^{\kappa^d}}
        \cdot C \cdot N\cdot\max\{1,\ell\}\\
       &\le C^2\,.
    \end{align*}
   Hence, statement (a$'$), i.e., $\neg\exists x\colon\varphi'$,
  is equivalent to statement (a).
\end{proof}

Corollaries~\ref{C1} and \ref{C2} allow to compute the truth value of
a closed formula $\varphi$ from $\FOMOD$ by, recursively, computing
the truth value of subformulas $\psi$ of $\varphi$ with arguments of
bounded size. More precisely, let $d$ be the quantifier depth of
$\varphi$ and set
\begin{equation}
  \label{eq:definition-D}
  D=2^{\max\Prod(\varphi)^{\kappa^{d+2}}}\cdot\max\CONST(\varphi)\,.
\end{equation}
Note that $\max\Prod(\varphi),\kappa\ge2$ implying $D\ge d$.

Now suppose $\exists x\colon\psi(x,y_1,\dots,y_\ell)$ is a subformula
of $\varphi$, $d'$ is the quantifier depth of $\psi$, and
$n_1,\dots,n_\ell$ are integers. Then to determine the truth of
$\exists x\colon\psi(x,n_1,\dots,n_\ell)$, it suffices by
Corollary~\ref{C1} to verify the truth of $\psi(n,n_1,\dots,n_\ell)$
for all integers $n$ with
  \begin{align*}
    |n| &\le 2^{\max\Prod(\psi)^{\kappa^{d'}}}\cdot\max\CONST(\psi)\cdot\max\{|n_1|,\dots,|n_\ell|\}\cdot\max\{1,\ell\}\\
        &\le 2^{\max\Prod(\varphi)^{\kappa^{d}}}\cdot\max\CONST(\varphi)\cdot\max\{|n_1|,\dots,|n_\ell|\}\cdot d\\
        &\le D\cdot\max\{|n_1|,\dots,|n_\ell|\}\,,
          \intertext{which, for later purposes, can be bounded by}
        &\le D^2\cdot\max\{|n_1|,\dots,|n_\ell|\}^2\,.
  \end{align*}
Similarly, suppose $\exists^{(q,p)} x\colon\psi(x,y_1,\dots,y_\ell)$
is a subformula of $\varphi$, $d'$ is the quantifier depth of $\psi$,
and $n_1,\dots,n_\ell$ are integers. Then to determine the truth of
$\exists^{(q,p)}x\colon\psi(x,n_1,\dots,n_\ell)$, it suffices by
Corollary~\ref{C2} to verify the truth of $\psi(n,n_1,\dots,n_\ell)$
for all integers $n$ with
  \begin{align*}
    |n| &\le \bigl(2^{\max\Prod(\psi)^{\kappa^{d'+1}}}\cdot\max\CONST(\psi)\cdot\max\{|n_1|,\dots,|n_\ell|\}\cdot\max\{1,\ell\}\bigr)^2\\
        &\le \bigl(2^{\max\Prod(\varphi)^{\kappa^{d+1}}}\cdot\max\CONST(\varphi)\cdot\max\{|n_1|,\dots,|n_\ell|\}\cdot d\bigr)^2\\
        &\le D^2\cdot\max\{|n_1|,\dots,|n_\ell|\}^2\,.
  \end{align*}
By induction, we obtain that all recursive calls of the evaluation
procedure use integers of size at most 
  \begin{align*}
    D^{4^d}
    &=\left(2^{\max\Prod(\varphi)^{\kappa^{d+2}}}\cdot\max\CONST(\varphi)\right)^{4^d}\\
    &\le
    2^{\max\Prod(\varphi)^{\kappa^{cd}}}\cdot\max\CONST(\varphi)^{4^d}\,,
  \end{align*}
where $c$ is some constant. To store any such integer, one needs space
$4^d\log D$. When evaluating a closed formula of quantifier depth $d$,
one has to store at most $d$ variables at once. Therefore we get the following.

\begin{prop}\label{P-FOMOD-}
  Satisfaction of a closed formula $\varphi\in\FOMOD$ of
  quantifier-depth $d$ can be decided in space $O(4^d\cdot \log D)$
  with $D$ given by Equation~\eqref{eq:definition-D}.
\end{prop}

Let $\varphi\in\FOMOD$. Then the quantifier depth $d$ is at most
$|\varphi|$. Since coefficients etc.\ are written in binary,
$\max\Prod(\varphi)$ and $\max\CONST(\varphi)$ are bounded by
$2^{|\varphi|}$. Consequently, the proposition shows that satisfaction
of closed formulas $\varphi\in\FOMOD$ can be decided in space doubly
exponential in $|\varphi|$.

Recall that for formulas from $\FOMOD$ we require modulo-counting
quantifiers of the form $\exists^{(t,p)}(y_1,\dots,y_\ell)$ to satisfy
$t\in\bN$ and $\ell=1$. We now show that also without this
restriction, the doubly exponential space bound remains true.

\begin{thm}
  Satisfaction of a closed formula $\varphi\in\FOMODmany$ can be decided
  in space doubly exponential in $|\varphi|$.
\end{thm}

\begin{proof}
  Let $\varphi\in\FOMODmany$ be a closed formula. By
  Proposition~\ref{P-FOMOD-to-FOMOD-}, we can compute in doubly exponential
  time an equivalent closed formula $\gamma\in\FOMOD$ without changing the
  sets of coefficients, moduli, or constants and without increasing
  the quantifier depth. Because of the time bound, this construction
  requires at most doubly exponential space and $|\gamma|$ is at most
  doubly exponential in $|\varphi|$.

  By Proposition~\ref{P-FOMOD-}, validity of $\gamma$ can be decided
  using space $O(4^{\qd(\gamma)}\cdot\log D)$ with $D$ given by
  Equation~\eqref{eq:definition-D}. Since $\gamma$ and $\varphi$ agree
  on the sets of coefficients etc.\ and on the quantifier depth, this
  value is doubly exponential in the size of $\varphi$.
\end{proof}

Since Proposition~\ref{P-CMOD-to-FOMOD} allows to translate formulas
from $\fullLogic$ into equivalent formulas from $\FOMODmany$ in
polynomial time, we also get the corresponding result for the logic
$\fullLogic$.

\begin{cor}
  Satisfaction of a closed formula $\varphi\in\fullLogic$ can be
  decided in space doubly exponential in $|\varphi|$.
\end{cor}

Note that this complexity matches the best known upper space bound for
Presburger arithmetic without modulo-counting quantifiers
from~\cite{FerR79}. From \cite{Ber80}, we know that Presburger
arithmetic can be decided in alternating doubly exponential time with
linearly many alternations (and our Corollary~\ref{C-complexity-of-C}
extends this to the logic $\C$). Our handling of the modulo-counting
quantifier requires us to count witnesses of bounded size. As the
number of potential witnesses is triply exponential, we do not see how
to do this using an alternating Turing machine in only doubly
exponential time.

\nocite{Pre91}

\bibliographystyle{alphaurl}

\bibliography{journal}
\end{document}